\documentclass[12pt,article]{amsart}
\usepackage{mathrsfs}
\usepackage{amssymb}
\usepackage{amsfonts}
\usepackage{amsbsy}
\usepackage{latexsym}
\usepackage{amssymb,latexsym,amsmath,amsthm}
\usepackage{framed}
\usepackage{diagbox}
\usepackage[colorlinks,linkcolor=blue]{hyperref}
\usepackage{graphicx}
\usepackage{xcolor}
\usepackage{epstopdf}
\usepackage{bm}
\usepackage{graphicx}
\theoremstyle{plain}

 \theoremstyle{remark} 

\renewcommand{\a}{\alpha}

\newtheorem {theo} {\bf Theorem} [section]
\newtheorem {prop} [theo] {\bf Proposition}

\newtheorem {lem} [theo] {\bf Lemma}
\newtheorem {note} [theo] {\bf Note}
\newtheorem {defi} {\bf Definition}[section]
\newtheorem{exam} {\bf Example}[section]

\newtheorem{rem}{\bf Remark}[section]
\numberwithin{equation}{section}
\usepackage{rotating}
\usepackage{mathdots}
\usepackage{threeparttable}
\usepackage{pgfplots}
\usepackage{graphicx}
\usepackage{subfigure}
\usepackage{caption}
\usepackage[normalem]{ulem}
\begin{document}
\title[Phase  retrieval of analytic signal]{\textcolor[rgb]{0.44,0.00,0.94}{The uniqueness of  phase   retrieval of analytic signals from very few STFT measurements}}
\author{Youfa Li}
\address{College of Mathematics and Information Science\\
Guangxi University,  Nanning, 530004, China }
\email{youfalee@hotmail.com}
\author{Hongfei Wang}
\address{College of Mathematics and Information Science\\
Guangxi University,  Nanning, 530004, China }
\email{whf254679@163.com}
\author{Deguang Han}
\address{
Department of Mathematics, University of Central Florida,  Orlando, FL 32816}
\email{Deguang.Han@ucf.edu}
\thanks{Youfa Li is partially supported by Natural Science Foundation of China (Nos: 61961003, 61561006, 11501132),  Natural Science Foundation of Guangxi (Nos: 2018JJA110110, 2016GXNSFAA380049) and  the talent project of  Education Department of Guangxi Government  for Young-Middle-Aged backbone teachers. Deguang Han  is  supported by the NSF grant   DMS-2105038.
}
%\subjclass[2010]{Primary 42C40; 65T60; 94A20}

\date{\today}

\begin{abstract}
Analytic signals  constitute a  class of signals that are  widely applied in  time-frequency analysis  such as
extracting instantaneous frequency (IF) or  phase derivative    in  the  characterization of  ultrashort laser pulse.  The purpose of this paper is to   investigate the phase retrieval (PR)  problem   for analytic signals in $\mathbb{C}^{N}$ by short-time Fourier transform (STFT) measurements since they enjoy some very nice structures. Since  generic  analytic signals   are   generally not sparse in the time domain, the existing
PR results for sparse (in time domain) signals do not apply to analytic signals. We will use bandlimited windows that usually have the full support length $N$ which allows us to get much  better    resolutions on low frequencies.
 More precisely, by exploiting the structure of  the STFT for analytic signals, we prove that the
STFT based phase retrieval (STFT-PR for short) of generic  analytic signals    can be achieved by their
$(3\lfloor\frac{N}{2}\rfloor+1)$ measurements. Since the generic  analytic signals
are     $(\lfloor \frac{N}{2}\rfloor+1)$-sparse in the  Fourier domain,
such a number of measurements is  lower than $4N+\hbox{O}(1)$ and $\hbox{O}(k^{3})$ which are  required
in the literature
for STFT-PR of all  signals  and of $k^{2}$-sparse (in the  Fourier domain) signals in  $\mathbb{C}^{N^{2}}$, respectively.
 Moreover, we also prove that  if the length  $N$ is even and  the windows are  also analytic,  then the number of measurements can be reduced to
 $(\frac{3 N}{2}-1)$.  As an application of this we get that the instantaneous frequency (IF)
 of a generic analytic signal  can be  exactly recovered  from the  STFT measurements.

%
%
%
%
%The existing results in the literature state that
%can be achieved by measurements.
%For the sparse case, it has been also  proved in the literature that
%measurements are required for the STFT-PR of    $k$-spare signals.
%Our main results state instead  that,   are sufficient for the PR of analytic signals. Such a number of measurement
%In optics, the  frequency-resolved optical gating (FROG) is widely used in characterizing the
%ultrafast pulse. On the other hand, analytic signals are commonly considered when    extracting the
%instantaneous features   of  events. In this paper we investigate the
%phase retrieval  problem of analytic signals by their STFT measurements.
%We first  established the ambiguity of
%the FROG-PR of analytic signals. Based on the ambiguity result, we found that the FROG-PR of analytic signals of even
%lengths   is essentially different from that of analytic signals of odd lengths. Moreover,  the existing
%approach of bandlimited signals holds for analytic signals of odd lengths but not for those of even lengths.
%We establish an approach to FROG-PR of analytic signals in $\mathbb{C}^{N}$ . It is proved  that such
%signals can be determined (up to the ambiguity) by their $(3\lfloor\frac{N}{2}\rfloor+1)$ STFT measurements through  the proposed approach.
\end{abstract}
\maketitle

\section{Introduction}\label{section1}
Phase retrieval (PR) is a nonlinear sampling   problem (c.f. \cite{Ba2,Ba3,FJR})  that asks  to
recover    a signal   $\textbf{z}\in \mathbb{C}^{N}$, \emph{up to  the potential ambiguity},  from   the magnitude  measurements
\begin{align}\notag \begin{array}{lll}  |\langle \textbf{z}, \textbf{a}_{k}\rangle|, \ k\in \Gamma, \end{array}\end{align}
where    $\textbf{a}_{k}\in \mathbb{C}^{N}$ is referred to as a measurement vector.
PR problem  is  of great interest    since it has been widely  applied in many applications
including   coherent diffraction imaging (CDI) (\cite{crystallography,2015Phase}),   quantum tomography
(\cite{Heinosaarri}), and holography (\cite{holographic}).
The most classical PR problem is to recover a signal  by its  Fourier transform measurements (\cite{Fienup,FJR}).

Associated with a  window $\textbf{w}\in \mathbb{C}^N$
and  a  separation parameter  $0<L<N$,
 the   short-time Fourier transform (STFT) or  Gabor transform of    a signal $\textbf{z}\in \mathbb{C}^N$   at $(k, m)$
is defined as (c.f. \cite{blind,LL}):
\begin{equation}\label{stftmeasuement}
\widehat{y}^{\textbf{w}}_{k,m}=\sum_{n=0}^{N-1} \textbf{z}_n\textbf{w}_{mL-n}e^{-2\pi\textbf{i} kn/N},
%\ k=0,1,\cdots,N-1, \ m=0,1,\cdots,\lceil\frac{N}{L}\rceil-1.
\end{equation}
where $k=0,1,\ldots,N-1$ and $m=0,1,\ldots, \lceil N/L\rceil-1$.
Compared with the  Fourier transform,   STFT is   more effective   for time-frequency localization   since   its associated    window enjoys great flexibility (c.f. \cite{waveletBOOK2,waveletBOOK3,waveletBOOK1,Matla}), and many deeper theoretical results related to Gabor frame analysis have been established in the literature (c.f. \cite{DS2016,ochenig,HANGU}).

%As such it  has  many applications in engineering problems such as  speech   processing and  ptychographical CDI.

%Since the phase information is missing in the PR problem,
 Finding the required number of measurements to do phase retrieval is always a fundamental issue, especially for practical applications  including  quantum tomography (c.f. \cite{Heinosaarri}).  For STFT-PR we refer to e.g. \cite{AR,BT,blind,TCN,Deguang,LL,Deguang1,2014GESPAR} for many recent results on this issue.
Bojarovska and Flinth \cite{BTF} characterized all the windows when $N^{2}$-number of  STFT measurements can recover all the  signals in $\mathbb{C}^{N}$.
%proved that $N^{2}$-{\color{blue} number of } STFT measurements are  sufficient for the recovery of any  signal in $\mathbb{C}^{N}$.
%What is more, by  such the  number of measurements sparse signals can be recovered   numerically stable
%(c.f. Shechtman,  Beck and   Eldar \cite{2014GESPAR}).
With the help of graph theory, Pfander
and  Salanevich \cite{Pfander} proved that the recovery of   any signal in $\mathbb{C}^{N}$ can be achieved by $\hbox{O}(N\log N)$ STFT measurements.
 From the perspective of frame theory (c.f. \cite{Deguang}),  the STFT-PR is essentially the PR problem by the frame measurement vectors  in $\mathbb{C}^N$.
There exist many  phase retrievable  frames of $4N+\hbox{O}(1)$-length (e.g. \cite{Balan,Xu1,Wangxu}).
As for the STFT-PR, the recovery can   be also  achieved by $4N+\hbox{O}(1)$ measurements (c.f. \cite{Balan,BTC,JK,ELDAR1}).
Note that the above mentioned  results hold for all the  signals in $\mathbb{C}^{N}$.
 By appropriately  choosing   the window $\textbf{w}$
and  the  separation parameter $L$,  Jaganathan, Eldar and  Hassibi \cite{JK} proved that  almost all non-vanishing signals in $\mathbb{C}^{N}$ can be determined  by
 their
 $3N+\hbox{O}(1)$ number of STFT measurements.
 Recently, the STFT-PR  for structured  signals
has attracted much attention  (e.g. \cite{BTF,EYC,JK}).
%Among other structured  signals, the class of  sparse (in the sense of  time or Fourier domain) signals are  such a typical example. \textcolor[rgb]{1.00,0.00,0.00}{With the a  priori knowledge} on sparsity, it is naturally  expected that
%fewer than  $4N+\hbox{O}(1)$ measurements are sufficient for the recovery.
In particular,  it was proved in    \cite{BTF}
that a $k^{2}$-sparse (in the  Fourier domain) signal  can be recovered  by $\hbox{O}(k^{3})$ number of STFT measurements.
In this paper we will investigate the phase retrieval problem for analytic signals that appear in many important applications such as
 time-frequency analysis (\cite{Cohen}),
 instantaneous frequency (IF) extracting in holography (e.g. \cite{Guo}),
and the characterization of  a changing pulse frequency (e.g. \cite{Gulley}).  As a proper subset of  $\mathbb{C}^{N}$, we are interested in finding fewer number of STFT measurements  than the above mentioned number
to guarantee the recovery  of any generic analytic signal.

The  definition of an  analytic signal was given by Marple \cite{analtyic}.
As in \cite{analtyic}  the  space $\mathbb{C}^{N}$ is supposed to  consist  of $N$-periodic and complex-valued signals  $\textbf{z}=(\textbf{z}_{0},
\ldots, \textbf{z}_{N-1})$ such that the subscripts are considered  modulo $N$.
%\subsection{Characterization of analytic signals}\label{tt}
For a real-valued signal $\textbf{x}\in \mathbb{R}^N$, its analytic signal $A(\textbf{x})=(A(\textbf{x})_{0},
\ldots, A(\textbf{x})_{N-1})$ is defined through its discrete Fourier transform  (DFT) $\widehat{A(\textbf{x})}=((\widehat{A(\textbf{x})})_0,\ldots,(\widehat{A(\textbf{x})})_{N-1})$, where   for  even length  $N$,
\begin{align}\label{up1}(\widehat{A(\textbf{x})})_k=
\left\{\begin{array}{lll}
\widehat{\textbf{x}}_{0},&k=0, \\
2\widehat{\textbf{x}}_{k},&1\le k \le \frac{N}{2}-1,\\
\widehat{\textbf{x}}_{\frac{N}{2}},&k=\frac{N}{2},\\
0,&\frac{N}{2}+1 \le k \le N-1,
\end{array}\right.
\end{align}
and for odd length  $N$,
\begin{align}\label{uwp1133} (\widehat{A(\textbf{x})})_k=
\left\{\begin{array}{lll}
\widehat{\textbf{x}}_{0},&k=0, \\
2 \widehat{\textbf{x}}_{k},&1\le k \le \frac{N-1}{2},\\
0,&\frac{N+1}{2} \le k \le N-1.
\end{array}\right.
\end{align}
From now on, the entire  set of  analytic signals on $\mathbb{C}^{N}$
is denoted by $\mathbb{C}_{A}^{N}$.
%By   \cite{analtyic}, the real part  $\Re(A(\textbf{x}))=\textbf{x}$ and the imaginary part  $\Im(A(\textbf{x}))$
%is the discrete Hilbert transform of $\textbf{x}$.

\begin{rem}\label{jiexixinhaodengjiax}
 By   \cite{analtyic}, we know that $\mathbb{C}_{A}^{N}=\{\textbf{x}+\textbf{i}H\textbf{x}: \textbf{x}\in \mathbb{R}^{N}\}$,  i.e.,
the real part  $\Re(A(\textbf{x}))=\textbf{x}$ and the imaginary part  $\Im(A(\textbf{x}))$
is $H\textbf{x}$, where $H$ is the discrete Hilbert transform.
%an analytic signal is uniquely determined by its real part.
%Moreover, the map $\mathcal{A}: \mathbb{R}^{N}\rightarrow \mathbb{C}^{N}$
%by $\mathcal{A}\textbf{x}=A(\textbf{x})$ is injective.
\end{rem}

We say that $ \textbf{z}\in \mathbb{C}^N$ is $B$-bandlimited  if its DFT  contains $N-B$ consecutive zeros.
For $\textbf{0}\neq\textbf{z}=(\textbf{z}_{0}, \ldots, \textbf{z}_{N-1})\in \mathbb{C}^N$,
  its support is defined to be  $\Xi=\big\{i\in\{0, \ldots, N-1\}: \textbf{z}_{i}\neq0\big\}$.
  Then the support length of $\textbf{z}$ is defined as  the cardinality $\#\Xi$.
 We also say that $\textbf{z}$ is $\#\Xi$-sparse.
  For a polynomial $f$ in $N$ (real or complex) variables, its   vanishing locus
 is  $V(f)=\{(x_{0}, \ldots, x_{N-1})\in \mathbb{R}^N (\hbox{resp.} \ \mathbb{C}^N)
: f(x_{0}, \ldots, x_{N-1})=0\}$. The complement of $V(f)$ in $\mathbb{R}^N$ (\hbox{resp.} $\mathbb{C}^N$)
is dense (c.f. \cite{blind}).
For  $x\in \mathbb{R}$, we will use the notation $\lceil x\rceil$ (respectively, $\lfloor x\rfloor$) to denote
the smallest (respectively, largest)  integer  that is not smaller (respectively, larger)  than $x$.

\subsection{Main result }
%For $N$ being even, we have the following.
%\begin{theo}\label{zhuyaojielunzongjie}
%Suppose that $N$ is even and  $\textbf{w}\in \mathbb{C}^N$  is a  structured  $B$-bandlimited window  for  STFT    such that $2\le B\le \frac{N}{2}+1$.
%Then for a generic   analytic signal $\textbf{z}\in \mathbb{C}^N$, its PR can be achieved by
%$(\frac{3N}{2}+1)$ STFT measurements.
%\end{theo}
We start with the definition of a generic analytic  signal  in $\mathbb{C}^{N}$.

\begin{defi}\label{kzxcvb}
When saying that a generic analytic  signal is uniquely determined by a collection of
polynomial measurements we mean that,  the    analytic  signals which cannot be determined by these
measurements  lie  in the vanishing locus of a nonzero polynomial on $\mathbb{C}^{N}$.
%This means
%that we can recover almost all generic  signals with the given measurements.
\end{defi}

The  main results  will be  stated in Theorems \ref{abc},  \ref{999} and \ref{3} which can be summarized as follows.
\begin{theo}\label{zhuyaojielunzongjie1}
Suppose  that  $\textbf{w}_{l}\in \mathbb{C}^N, l=1, \ldots, M$  are the   structured  $B$-bandlimited windows  for  STFT    such that
 their bandlimits $2\le B\le \lceil\frac{N}{2}\rceil+1$. Moreover, the STFT separation
% (between adjacent short-time sections)
  parameter $0<L<N$
satisfies  $\lceil N/L\rceil\geq3$.
Then for a generic   analytic signal $\textbf{z}\in \mathbb{C}^N$, it  can be recovered (up to a sign)  by its
$(3\lfloor\frac{N}{2}\rfloor+1)$  number of STFT measurements. Moreover, if the length $N$ is even and the windows are analytic then the above number
of measurements can be reduced to $(\frac{3 N}{2}-1)$.
%\sout{Herein $\lceil x\rceil$ ($\lfloor x\rfloor$)
%is   the smallest (largest)  integer  that is not smaller (larger)  than $x$.}
\end{theo}

%{\color{red} Youfa, I am a little bit confused by the following proposition. Are you saying: The IF of very generic analytic signal can be exactly recovered by the number of measurements outlined in Theorem 1.1?  Do we need to define  instantaneous frequency?  Try not to use ``easy" unless it is really obvious. We might need to provide (even a short one) proof for it }: \textbf{Professor Han, the definition of IF is added
%as follows, and the proof for Proposition \ref{FGHH} is also provided.}

The STFT in  Theorem \ref{zhuyaojielunzongjie1} requires multiple bandlimited windows.
The following concerns the application background of such a type of STFT.

\begin{rem}\label{connection}
(1) The window $\textbf{w}_{l}$ in Theorem \ref{zhuyaojielunzongjie1} is bandlimited.
By the discrete  uncertainty principle (c.f. \cite[section 3.2]{BTF})   its  support length is generally
$N$. That is, $\textbf{w}_{l}$ is generally  a long window. By \cite{Mateo,Pihlaj},
longer windows on low frequencies allow getting better frequency resolution,  and
 they have been  used  in some STFT-PR approaches (e.g.   \cite{HAMIDNAWAB}).
%If the  bandlimit of the window  is less than $(\lceil\frac{N}{2}\rceil+1)$, then the PR can be achieved by only one window function (Theorem  \ref{abc}). However, for the case that windows \textcolor[rgb]{1.00,0.00,0.00}{are} $(\lceil\frac{N}{2}\rceil+1)$-bandlimited,
%it will be explained in Remark \ref{diergeyaoqiu} and Example \ref{fanlishuoming} that
%the PR may not be achieved by  single-window measurements. Therefore,
%the multiple-window measurements are  required in Theorem   \ref{999} for windows being $(\lceil\frac{N}{2}\rceil+1)$-bandlimited.
(2) Multiple-window measurements were used in Theorem \ref{zhuyaojielunzongjie1}, and such a type of STFT measurements were  also used for STFT-PR in
\cite{LL}. We point out that multiple-window  approach is particularly useful in  coded diffraction patterns (c.f. \cite{Candes}).
%Through STFT,
%it is not possible to have arbitrarily time frequency resolution simultaneously  with good frequency resolution
%(c.f. \cite{waveletBOOK2}). However,
\end{rem}

\begin{rem}\label{remark2345}
(1) Since   any generic  analytic signal   is   $(\lfloor \frac{N}{2}\rfloor+1)$-sparse  in the  Fourier domain,
as mentioned previously     it  is generally not sparse in the time domain. Therefore, the existing
PR results for sparse (in time domain) signals do not hold for analytic signals. (2)  A generic  analytic signal
is $(\lfloor \frac{N}{2}\rfloor+1)$-sparse (in the  Fourier domain) or equivalently has   bandlimit  $B=\lfloor \frac{N}{2}\rfloor+1$.  Theorem \ref{zhuyaojielunzongjie1} implies   that it can be determined   (up to  a sign)
by its  $(3B-2)$ STFT measurements. When the windows are analytic, such a  required  number of measurements
can be reduced to $(3B-4)$.
%\textcolor[rgb]{1.00,0.00,0.00}{(3) As for the case of windows of larger  bandlimit, it will be \textcolor[rgb]{1.00,0.00,0.00}{explained}
%in subsection \ref{examples1} that  more measurements may be  required to do the PR.}
%Moreover, the numerical examples in section \ref{numericalexample} will  imply  that
%Approach \ref{APPROACH11} (using the narrow bandlimited windows) outperforms
%Approaches \ref{APPROACH21} and \ref{APPROACH31}  (using the wide bandlimited windows).
\end{rem}

An immediate consequence of Theorem \ref{zhuyaojielunzongjie1} is the exact recovery of instantaneous frequency  (IF)  for generic
analytic signals.  %It can be  \sout{easily}  proved by its sign ambiguity of the recovery. ({\color{red} No need to say this if we provide a proof})
 Given  an analytic signal $\textbf{z}=(\textbf{z}_{0}, \ldots, \textbf{z}_{N-1})$, denote
  its element $\textbf{z}_{k}$  by $|\textbf{z}_{k}|e^{\textbf{i}\arg(\textbf{z}_{k})}$
with $\arg(\textbf{z}_{k})\in [0, 2\pi)$.
Define
\begin{equation}\label{xiangweidaoshudingyweri} \varphi^*(k):=(\arg(\textbf{z}_{k})-\arg(\textbf{z}_{k-1}))\mod 2\pi.\end{equation}
Then $\textbf{z}_{\varphi}:= (\varphi^*(0),\ldots,\varphi^*(N-1))$ is referred to as the phase derivative (PD) or IF of $\textbf{z}$ (c.f.\cite{Dang2}).

\begin{prop}\label{FGHH}
Suppose that $\textbf{z}\in \mathbb{C}^{N}$ is a generic analytic signal.
Then its IF can be exactly recovered from the  same number of STFT measurements  as specified in Theorem \ref{zhuyaojielunzongjie1}.
\end{prop}
\begin{proof}
By Theorem \ref{zhuyaojielunzongjie1}, we get  $\textbf{z}$ or $-\textbf{z}$.
Since the $k$-th element of  $-\textbf{z}$ is expressed as
$|\textbf{z}_{k}|e^{\textbf{i}((\arg(\textbf{z}_{k})-\pi)\mod 2\pi)},$
 $[(\arg(\textbf{z}_{k})-\pi)\mod 2\pi-(\arg(\textbf{z}_{k-1})-\pi)\mod 2\pi]\mod 2\pi=(\arg(\textbf{z}_{k})-\arg(\textbf{z}_{k-1}))\mod 2\pi.$
That is, the IF of $-\textbf{z}$ is identical to that of
$\textbf{z}$.  This completes the proof.
\end{proof}

%\subsection{Connections with the STFT-PR results in the literature}

\subsection{Comparisons with the existing  results}
In this subsection we make some comparisons between Theorem \ref{zhuyaojielunzongjie1}
and the results in \cite{BTF,lima,JK}.

%%\sout{The following remark gives an interpretation of Theorem \ref{zhuyaojielunzongjie1}
%%from correlation between  the number of measurements and the bandlimits of windows.}
%{\color{red} Are you able to make connections with some papers in the literature?} \textbf{Professor Han,
%we give two connections in the following remark. Please check.  As for the bandlimited aspect,  the bandlimited windows are bandlimited  in our paper. We found the bandlimited window
%for STFT was also used in  this paper  [A frequency domain analysis of the error distribution from
%noisy high-frequency data, Biometrika, 2018]. But this paper concerns on  non-periodic signals.
%So do you think we need to give a connection to this paper?}

%We start with the connection of our results with those in the literature.

%{\color{red} Youfa, the above remark does add much value to the significancy of the theorem, it probably causes more confusion since ... }

% In what follows, we comment on the comparisons of Theorem \ref{zhuyaojielunzongjie1} with the   results in  \cite{BTF} and  \cite{lima}.
\begin{note}
As previously mentioned, it was proved by \cite{BTF} that  the STFT-PR of  a $k^{2}$-sparse (in the  Fourier domain) signal in $\mathbb{C}^{N^2}$
can be achieved by $O(k^{3})$ measurements.   Note that
a generic analytic signal is  $(\lfloor \frac{N}{2}\rfloor+1)$-sparse. Theorem \ref{zhuyaojielunzongjie1}
implies that its STFT-PR can be achieved  by  using only $3\lfloor\frac{N}{2}\rfloor+O(1)$ number of measurements, which is a significant improvement when restricting to analytic signals.
\end{note}

%\sout{The following remark concerns on the comparison between  Theorem \ref{zhuyaojielunzongjie1}
%and FROG-PR results in \cite{lima}.}

\begin{note}
Frequency-resolved optical gating (FROG) trace  is essentially  an adaptive STFT since the corresponding
window is the delay of the signal itself (c.f. \cite{2017On}). The FROG-PR for analytic signals was addressed in \cite{lima}. The STFT-PR
in this paper is different from FROG-PR since the window here is known and independent of  the signal.  The following tells us some other essential differences between  Theorem \ref{zhuyaojielunzongjie1} and the results in \cite{lima}.  (1) The main result in \cite{lima} %\sout{just holds for  the  length   $N$ being even,}
only applies to the case when $N$ is even  and
the separation  parameter $0<L<N$ is
%\sout{required to be}
 odd with the property that $\lceil N/L\rceil\geq5$.
However,  Theorem \ref{zhuyaojielunzongjie1} only requires $\lceil N/L\rceil\geq3$ in this case,
and the odevity of $L, N$ is not required.
% \sout{that still holds for
% $N$ and $L$ being odd and  even, respectively, as long as  $\lceil N/L\rceil\geq3$. }
 %(2)  Similarly,
% $(\frac{3N}{2}+1)$ FROG measurements are required in \cite{lima} for the even length case. The number of measurement  can be  reduced to $(\frac{3N}{2}-1)$ in Theorem \ref{zhuyaojielunzongjie1}.
%\sout{ Instead,
% Theorem \ref{zhuyaojielunzongjie1} states that such a number can be reduced to $(\frac{3N}{2}-1)$.}
 (2) The ambiguity for FROG-PR in \cite{lima} is
% \sout{much more  complicated than}
 different from  that in Theorem \ref{zhuyaojielunzongjie1}
 since it   additionally  contains  shift and reflection. Consequently, the IFs of  only a few analytic signals can be extracted from
 FROG-PR (\cite[section 4]{lima}).  However,  Proposition  \ref{FGHH} applies to every generic analytic signal.
% \sout{ implies that from the aspect of ambiguity and IF recovery  our results here is also essentially different from those in
% \cite{lima}.}
\end{note}

\begin{note}
Recall again that the PR result  in \cite{JK}
holds for almost all non-vanishing  signals in $\mathbb{C}^{N}$
by polynomial measurements. The complement of the set of almost all non-vanishing signals
has the measure zero.
%This  means that the set of signals which can not be determined by these
%measurements
%has the measure zero.
On the other hand,
it follows from Remark \ref{remark2345} that
for $N\geq2$  the analytic signal in $\mathbb{C}^{N}$ is bandlimited.
 Since the  set of analytic signals has  measure zero,  Theorem \ref{zhuyaojielunzongjie1}\ does not contradicts with
the result in  \cite{JK}. We  will address this a little bit more in the next section.

\end{note}

\section{Preliminary}
A complex number  $0\neq z\in \mathbb{C}$  is traditionally  denoted by  $|z|e^{\textbf{i}\arg(z)}$, where    $\textbf{i}$, $|z|$ and  $\arg(z)$ are   the imaginary unit,   modulus and     phase, respectively.
The   real and imaginary parts, and  conjugation of $z$ are denoted by
$\Re(z)$,$\Im(z)$ and $\bar{z}$, respectively.

The discrete Fourier transform  (DFT) of $\textbf{z} \in \mathbb{C}^N$ is defined  by
$\widehat{\textbf{z}}:=(\widehat{\textbf{z}}_0,\widehat{\textbf{z}}_1,\ldots,\widehat{\textbf{z}}_{N-1})$
such that
$\widehat{\textbf{z}}_{k}= \sum_{n=0}^{N-1}\textbf{z}_{n}e^{-2\pi \textbf{i}kn/N}.$
The inverse discrete Fourier transform (IDFT) admits
the formula  \begin{align}\label{cba} \textbf{z}_{k}= \frac{1}{N}\sum_{n=0}^{N-1}\widehat{\textbf{z}}_{n}e^{2\pi \textbf{i}kn/N}.\end{align}
By the  IDFTs of $ \textbf{z}$ and $\textbf{w}$,   the STFT $\widehat{y}^{\textbf{w}}_{k,m}$
in \eqref{stftmeasuement} can be expressed as
 \begin{equation}\label{stft}
\widehat{y}^{\textbf{w}}_{k,m} =\dfrac{1}{N}\sum_{l=0}^{N-1}\widehat{\textbf{z}}_{k+l}\widehat{\textbf{w}}_l\omega^{lm},
\end{equation}
where $\omega=e^{\frac{2\pi\textbf{i}L}{N}}$.
The following is a   characterization of the  DFT structure for an analytic signal.
\begin{prop}\label{jiegou}
 (c.f. \cite{lima}) Suppose that $\textbf{z}\in \mathbb{C}^N$. Denote the Cartesian product of sets by $\times$. Then $\textbf{z}$ is analytic if and only if the following two items holds:

(i) for even length $N$, $\widehat{\textbf{z}}\in \mathbb{R}\times \mathbb{C}^{\frac{N}{2}-1}\times \mathbb{R}\times \overbrace{\{0\}\times \ldots\times \{0\}}^{(\frac{N}{2}-1)\text{copies}}$;

(ii) for odd length $N$, $\widehat{\textbf{z}}\in \mathbb{R}\times \mathbb{C}^{\frac{N-1}{2}}\times \overbrace{\{0\}\times\ldots\times \{0\}}^{\frac{N-1}{2}\text{copies}}$.
\end{prop}

 %Next we introduce the definition of a generic analytic signal.

%\textcolor[rgb]{0.44,0.00,0.94}{\begin{defi}\label{ddcvx}
%We say that a  set $\Theta\subseteq \mathbb{C}_{A}^{N}$
%is generic in $\mathbb{C}_{A}^{N}$ if its complement  in $\mathbb{C}_{A}^{N}$
%%When saying that a generic analytic  signal in $\mathbb{C}^{N}$ is uniquely determined by a collection of
%%polynomial measurements, we mean that the set of analytic  signals $\{A(\textbf{x}): \textbf{x}\in \Lambda\subseteq \mathbb{R}^{N}\}$ which cannot be determined by these
%%measurements
%lies in the vanishing locus  of a nonzero polynomial on $\mathbb{C}^{N}$.
%\end{defi}}
The following gives a characterization of the generic analytic signals.
\begin{prop}\label{ncer}
Let  $\Theta$ be a  set of  generic  analytic signals in $\mathbb{C}^{N}$ such that     any  signal  in     $\Theta$   can be determined by a collection of polynomial  STFT measurements. Meanwhile,   all the
signals in the  complement $\mathbb{C}_{A}^{N}\setminus\Theta$ can   not be determined by these  measurements and they lies in the vanishing locus of a nonzero polynomial  $f$ on $\mathbb{C}^{N}$. Denote $\mathbb{C}_{A}^{N}\setminus\Theta$  by
$\{A(\textbf{x})=\textbf{x}+\textbf{i}H\textbf{x}: \textbf{x}\in \Lambda\subseteq \mathbb{R}^{N}\}$.
Then $\Lambda$
lies in the vanishing locus of a nonzero polynomial $g$ on $\mathbb{R}^{N}$.
\end{prop}
\begin{proof}
For any analytic signal $A(\textbf{x})$,
it follows from Remark \ref{jiexixinhaodengjiax} that $H\textbf{x}=\frac{A(\textbf{x})-\textbf{x}}{\textbf{i}}$.
%,
%where $H$ is the discrete Hilbert transform.
Consequently,
$\widehat{H\textbf{x}}=\frac{\widehat{A(\textbf{x})}-\widehat{\textbf{x}}}{\textbf{i}}$.
By \eqref{up1} and \eqref{uwp1133},
 for any  $\textbf{x}, \textbf{y}\in \mathbb{R}^{N}$ we have
$\widehat{H(\textbf{x}+\textbf{y})}=\widehat{H\textbf{x}}+\widehat{H\textbf{y}}$.
From this and the linearity of IDFT
we have that $H(\textbf{x}+\textbf{y})=H\textbf{x}+H\textbf{y}$. That is, the
discrete Hilbert transform $H$ is linear.
Then there exists a  matrix $H_{\hbox{e}}\in \mathbb{R}^{N\times N}$
such that  $H\textbf{x}=H_{\hbox{e}}\textbf{x}$ for any $\textbf{x}\in \mathbb{R}^{N}$.
Consequently,  \begin{align}\label{JNVBCX}A(\textbf{x})=(I+\textbf{i}H_{\hbox{e}})\textbf{x},
\end{align}
where $I$ is the identity matrix.   For any $\textbf{x}+\textbf{i}H\textbf{x}\in \mathbb{C}_{A}^{N}\setminus\Theta$ it
follows   from \eqref{JNVBCX} that
$f(\textbf{x}+\textbf{i}H\textbf{x})=f((I+\textbf{i}H_{\hbox{e}})\textbf{x})=0.$ By choosing $g(\textbf{x}):=f((I+\textbf{i}H_{\hbox{e}})\textbf{x})$,
the proof is completed.
\end{proof}

\section{Main results}

\subsection{Several lemmas}
This section  starts  with an auxiliary result from \cite[Lemma 3.2]{2017On}.
\begin{lem}\label{L}
Consider  an equation system w.r.t $z\in \mathbb{C}$:
\begin{equation}\label{8790}
\left\{
\begin{aligned}
|z+v_1| &=n_1,\\
|z+v_2| &=n_2,\\
|z+v_3| &=n_3,
\end{aligned}
\right.
\end{equation}
where  $v_1,v_2,v_3\in \mathbb{C}$ are distinct. If there exists a solution $\tilde{z}=a+\textbf{i}b$ to the above system
and  $\Im\big(\dfrac{v_1-v_2}{v_1-v_3}\big)\ne 0$, then it is the unique one. Moreover, it is given by
$$\left(\begin{array}{cccccccccc}a\\
b
\end{array}\right)=\frac{1}{2}\left(\begin{array}{cccccccccc} c &d \\e &f\end{array}\right)^{-1} \left(\begin{array}{cccccccccc}n_1^2-n_2^2+|v_2|^2-|v_1|^2\\n_1^2-n_3^2+|v_3|^2-|v_1|^2\end{array}\right),$$
where $c=\Re(v_1-v_2)$, $d=\Im(v_1-v_2)$, $e=\Re(v_1-v_3)$ and  $f=\Im(v_1-v_3)$.
\end{lem}

The following two lemmas will be needed in the proofs of  Theorems \ref{abc} and \ref{999}.

\begin{lem}\label{lemma3.2}
Suppose that   a  rational function  $f(z)=\dfrac{az+b}{cz+d}$, $z=x+\textbf{i}y\in \mathbb{C}$    satisfies the conditions  $ad-bc\ne 0,ac\ne 0$,
$a,b,c,d\in\mathbb{C}$. Then the set $\{(x,y)\in \mathbb{R}^{2}: \Im(f(z))=0, z=x+\textbf{i}y\}$
lies in the vanishing locus of a nonzero polynomial on $\mathbb{R}^{2}$.
%for a \textcolor[rgb]{0.00,0.07,1.00}{generic $z$ ?}, it holds that $\Im(f(z))\neq0$.
%the  Lebesgue measure $E(\mathcal{Z}_{f_{\Im}})=0$ where
%$\mathcal{Z}_{f_{\Im}}$ is the zero set of .
\end{lem}

\begin{proof}
Let $z=x+\textbf{i}y$.
We have
%$f(z)=\dfrac{az+b}{cz+d}=\dfrac{(az+b)(\bar{c}\bar{z}+\bar{d})}{|cz+d|^2}$ that
 \begin{align} \label{08765} \Im(f(z))&=\Im(\frac{(az+b)(\bar{c}\bar{z}+\bar{d})}{(cz+d)(\bar{c}\bar{z}+\bar{d})})=\Im(\frac{a\bar{c}|z|^2+a\bar{d}z+b\bar{c}\bar{z}+b\bar{d}}{|cz+d|^2}) \nonumber\\
 &=\dfrac{\Im(a\bar{c})(x^2+y^2)+(\Im(a\bar{d})+\Im(b\bar{c}))x+(\Re(a\bar{d})-\Re(b\bar{c}))y+\Im(b\bar{d})}{|cz+d|^2}.\end{align}
%$f(z)=\dfrac{az+b}{cz+d}=\dfrac{(az+b)(\bar{c}\bar{z}+\bar{d})}{|cz+d|^2}$ that
 %\begin{align} \label{08765} \Im(f(z))=\dfrac{\Im(a\bar{c})(x^2+y^2)+(\Im(a\bar{d})+\Im(b\bar{c}))x+(\Re(a\bar{d})-\Re(b\bar{c}))y+\Im(b\bar{d})}{|cz+d|^2}.\end{align}
 Now the proof can be    completed by choosing the polynomial $F(x,y):=
 \Im(a\bar{c})(x^2+y^2)+(\Im(a\bar{d})+\Im(b\bar{c}))x+(\Re(a\bar{d})-\Re(b\bar{c}))y+\Im(b\bar{d}).$
%
%\textbf{Case I}: $\Im(a\bar{c})=0$. Then $\Im(f(z))=\dfrac{(\Im(a\bar{d})+\Im(b\bar{c}))x+(\Re(a\bar{d})-\Re(b\bar{c}))y+\Im(b\bar{d})}{|cz+d|^2}$,
%and $\mathcal{Z}_{f_{\Im}}$ lies on  a plane of $\mathbb{R}^{2}$. Clearly, for the generic $z$  it holds that $\Im(f(z))\neq0$.
%
%%
%%
%%
%%(1)??|??????????????????a??????a????????????????????3???3?????????????????━??━???3??????a?????????a$E(Z_h)=0$.
%\textbf{Case II}: $\Im(a\bar{c})\ne0$.
%%Define  \begin{align}\label{pq1}
%%\left\{
%%\begin{aligned}
%%u &:=\frac{1}{|cz+d|^2}\Im(a\bar{c})(x^2+y^2), (\ref{pq1} A)\\
%%v &:=\frac{-1}{|cz+d|^2}[(\Im(a\bar{d})+\Im(b\bar{c}))x+(\Re(a\bar{d})-\Re(b\bar{c}))y+\Im(b\bar{d})], (\ref{pq1} B).
%%\end{aligned}
%%\right.
%%\end{align}
%Then  $\{z: \Im(f(z))=0\}$ is  the zero set of a polynomial.
%This completes the proof.
\end{proof}

\begin{lem}\label{m}  Let $L$ and $ N$ be  such that $0<L<N$ and $\lceil N/L\rceil\geq3$.
If $0\le m_1,m_2,m_3\le \lceil N/L\rceil-1$ are distinct, then $\Im\big(\dfrac{\omega^{m_1}-\omega^{m_2}}{\omega^{m_1}-\omega^{m_3}}\big)\ne0$ where $\omega=e^{\frac{2\pi\textbf{i}L}{N}}$.
\end{lem}

\begin{proof}
Since
$$\dfrac{\omega^{m_1}-\omega^{m_2}}{\omega^{m_1}-\omega^{m_3}}=\dfrac{(\omega^{m_1}-\omega^{m_2})(\overline{\omega^{m_1}-\omega^{m_3}})}{|\omega^{m_1}-\omega^{m_3}|^2}=
\dfrac{1-\omega^{m_1-m_3}-\omega^{m_2-m_1}+\omega^{m_2-m_3}}{|\omega^{m_1}-\omega^{m_3}|^2},$$ we get that the condition
% it follows that
$\Im(\frac{\omega^{m_1}-\omega^{m_2}}{\omega^{m_1}-\omega^{m_3}})\ne0$ is equivalent to  the condition
\begin{align}\label{90955}\begin{array}{llll}
\Im(1-\omega^{m_1-m_3}-\omega^{m_2-m_1}+\omega^{m_2-m_3})\\
=-\sin\Big(\frac{2\pi(m_1-m_3)L}{N}\Big)-\sin\Big(\frac{2\pi(m_2-m_1)L}{N}\Big)+\sin\Big(\frac{2\pi(m_2-m_3)L}{N}\Big)\ne0.
\end{array}\end{align}

Assume to the contrary that $\Im(\frac{\omega^{m_1}-\omega^{m_2}}{\omega^{m_1}-\omega^{m_3}}) =0$. Then  we have $$\sin\Big(\frac{2\pi(m_2-m_3)L}{N}\Big)=\sin\Big(\frac{2\pi(m_1-m_3)L}{N}\Big)+\sin\Big(\frac{2\pi(m_2-m_1)L}{N}\Big),$$   which implies that
%???????????????━??━??????????????━??━??????a?o????????$-sin(\frac{2\pi(m_1-m_3)}{r})-sin(\frac{2\pi(m_2-m_1)}{r})+sin(\frac{2\pi(m_2-m_3)}{r})=0$, ??|??????|????$$sin(\frac{2\pi(m_2-m_3)}{r})=sin(\frac{2\pi(m_1-m_3)}{r})+sin(\frac{2\pi(m_2-m_1)}{r})$$, ????????????????????????1???????o???o???2??????????o????????1????????|??????|????
 $$\sin\Big(\!\frac{\pi(m_2\!-\!m_3)L}{N}\Big)\!\cos\Big(\!\frac{\pi(m_2\!-\!m_3)L}{N}\Big)\!=\!\sin\Big(\!\frac{\pi(m_2\!-\!m_3)L}{N}\Big)\!\cos\Big(\!\frac{\pi(2m_1\!-\!m_3\!-\!m_2)L}{N}\Big).$$
Since $m_2\ne m_3$ and $0\le m_2, m_3\le \lceil N/L\rceil-1$,
%\sout{then}
 we get that
 $\sin\Big(\frac{\pi(m_2-m_3)L}{N}\Big)\ne 0$ and hence  $$ \cos\Big(\frac{\pi(m_2-m_3)L}{N}\Big)=\cos\Big(\frac{\pi(2m_1-m_3-m_2)L}{N}\Big).$$
This implies that  $\frac{\pi(m_2-m_3)L}{N}=\frac{\pi(2m_1-m_3-m_2)L}{N}$ or $\frac{\pi(m_2-m_3)L}{N}=-\frac{\pi(2m_1-m_3-m_2)L}{N}$. Thus we have  either $m_1=m_2$ or $m_1=m_3$, which leads to  a contradiction. The proof is completed.
\end{proof}

\subsection{The first main result:    window bandlimit $2\le B\le \lceil\frac{N}{2}\rceil$ case}\label{firstmainresult}
%%%% 3/5

Suppose that   the window  $\textbf{w}\in \mathbb{C}^N$ is $B$-bandlimited
such that  $2\le B\le \lceil\frac{N}{2}\rceil$. Consequently,  there exists
$i\in \{0, \ldots ,N-1\}$ such that \begin{align}\label{tiaojian} \widehat{\textbf{w}}_i=\cdots=\widehat{\textbf{w}}_{i+N-B-1}=0, \widehat{\textbf{w}}_{i+N-B}\neq0.\end{align}
For such a  subscript  $i$, we consider the following measurements
\begin{align}\label{fangcheng12} \big|\widehat{y}^{\textbf{w}}_{\lfloor\frac{N}{2}\rfloor-(i+N-B)-n+1,m}\big|=\dfrac{1}{N}\big|\sum_{l=0}^{N-1}\widehat{\textbf{z}}_{\lfloor\frac{N}{2}\rfloor-(i+N-B)-n+1+l}\widehat{\textbf{w}}_l\omega^{lm}\big|, \end{align}
where  $n=1, \ldots,\lfloor\frac{N}{2}\rfloor+1.$
%where $k_{n}=\lfloor\frac{N}{2}\rfloor-(i+N-B)-n+1.$
%$\lfloor\frac{N}{2}\rfloor-(i+N-B)=k_1>k_2>\cdots>k_{\lfloor\frac{N}{2}\rfloor+1}=N-i+B$

The following is a $2$-bandlimited window in $\mathbb{C}^{48}$
such that \eqref{tiaojian} holds with $i=2.$

\begin{exam}
For $N=48$  we  design  a window $\textbf{w}$
such that its bandlimit   $B=2$.  Its real and imaginary parts are plotted in Figure \ref{zzz} (a) while  the real and imaginary parts of $\widehat{\textbf{w}}$ are plotted in Figure \ref{zzz} (b).
\end{exam}
\begin{figure}[htbp] %?2????━o???━??━????????????
\centering %????????????????
\subfigure[$\textbf{w}$]
{\begin{minipage}{7cm} %?2????━o???━??━????????????
\centering
\includegraphics[scale=0.5]{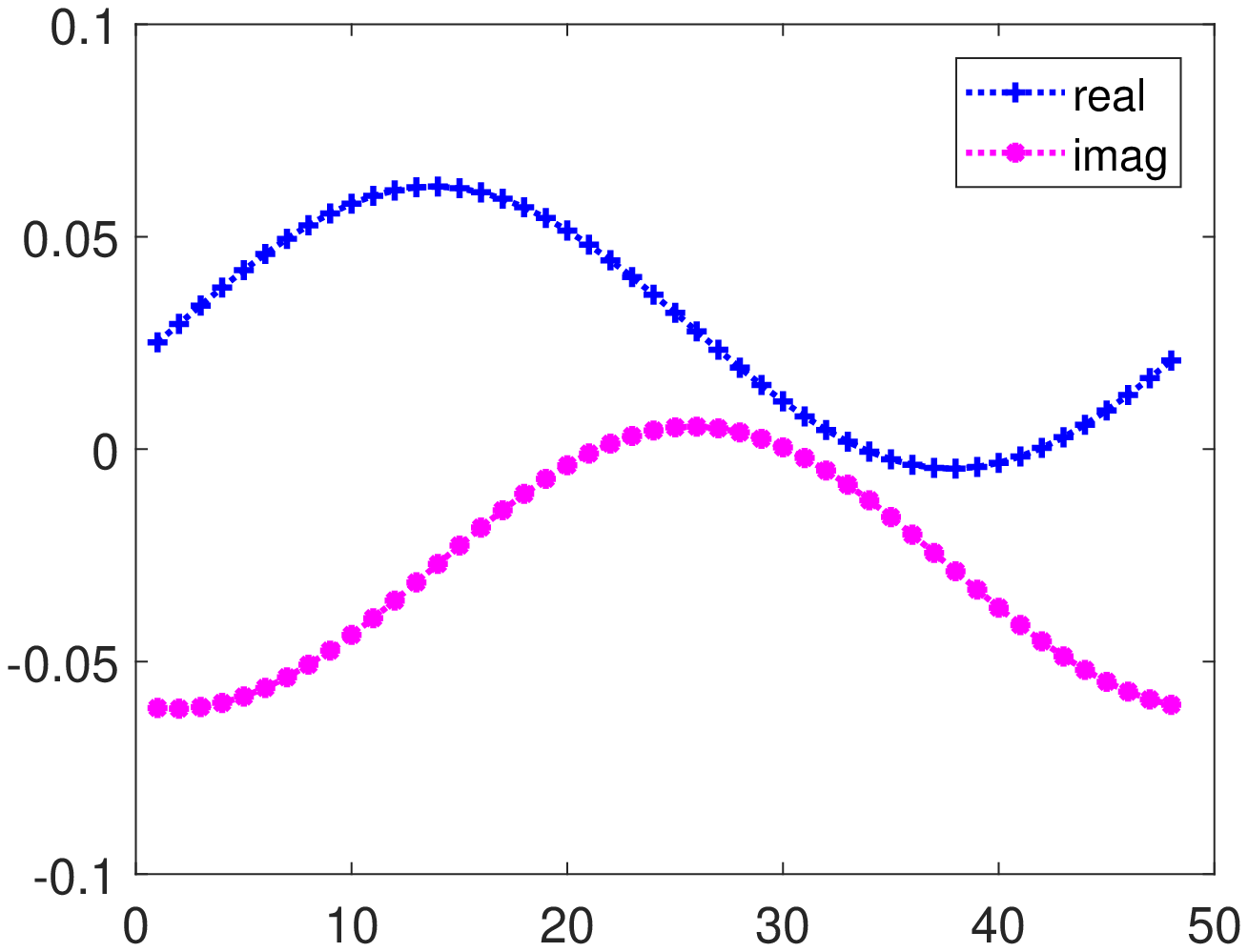}
\end{minipage}}
\subfigure[$\widehat{\textbf{w}}$]
{\begin{minipage}{7cm} %?2????━o???━??━????????????
\centering
\includegraphics[scale=0.5]{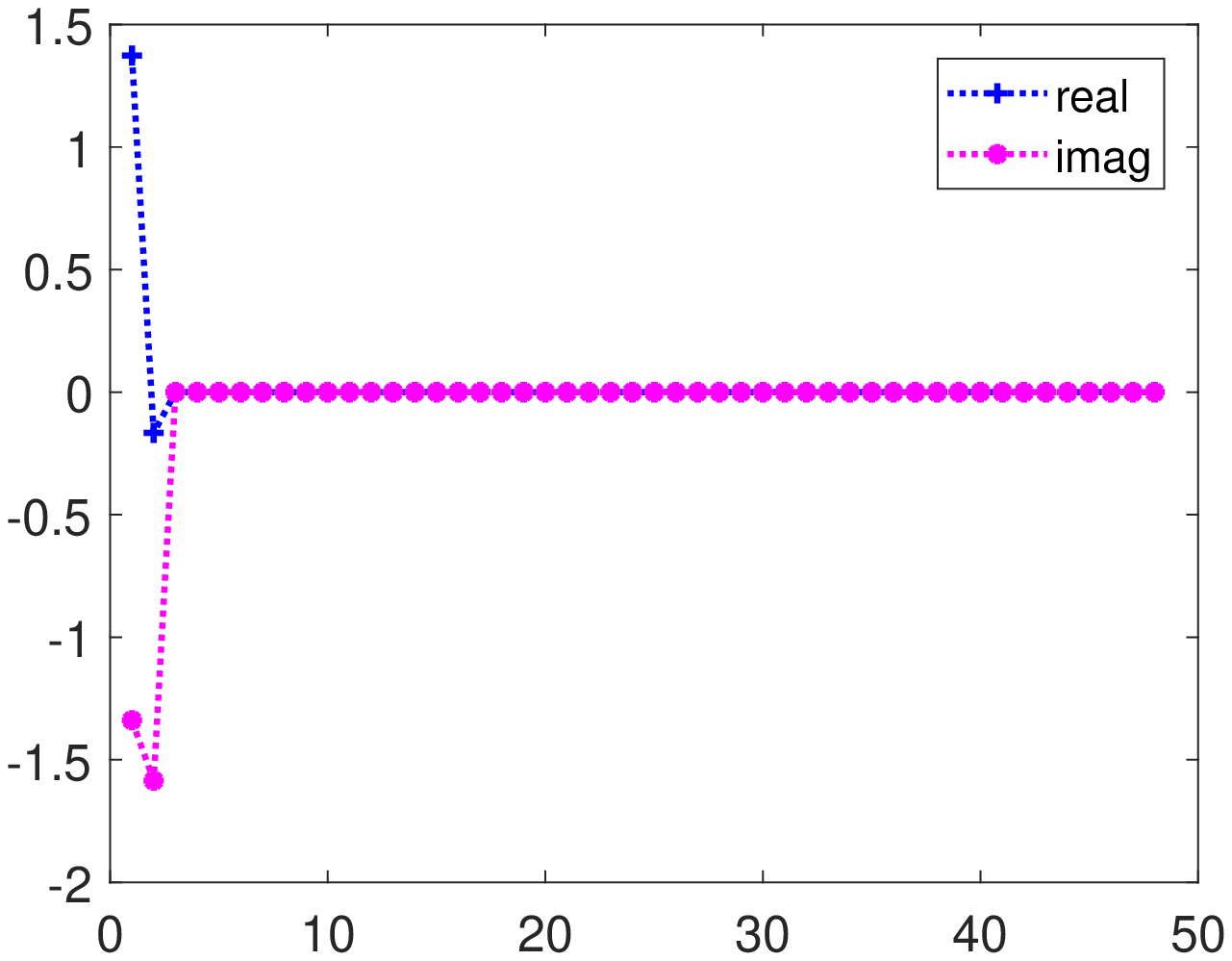}
\end{minipage}}
\caption{(a) The graph of   the real and imaginary parts of  $\textbf{w}$;   (b) The graph of   the real and imaginary parts  of  $\widehat{\textbf{w}}$.}
\label{zzz}
\end{figure}

 The following example on the  the summation in \eqref{fangcheng12}
for  the case $N=6$ and $B=3$ exhibits the  structure of
$|\widehat{y}^{\textbf{w}}_{\lfloor\frac{N}{2}\rfloor-(i+N-B)-n+1,m}|$
for  general $N$ and $B$.

\begin{exam}
Let $N=6,B=\lceil\frac{N}{2}\rceil=3.$ For an analytic signal $\textbf{z}\in \mathbb{C}^{6}$,
it follows from Proposition   \ref{jiegou} (i) that
its   DFT $\widehat{\textbf{z}}=(\widehat{\textbf{z}}_0,\widehat{\textbf{z}}_1,\widehat{\textbf{z}}_2,\widehat{\textbf{z}}_3, 0,  0)$.
Choose a $3$-bandlimited window $\textbf{w}\in \mathbb{C}^{6}$ such that
$\widehat{\textbf{w}}=(\widehat{\textbf{w}}_0,\widehat{\textbf{w}}_1,\widehat{\textbf{w}}_2,0,0,0)$
 and correspondingly  $i=3$ in \eqref{tiaojian}.
 For $n=1, \ldots, 4$, $|\widehat{y}_{4-n,0}^{\textbf{w}}|$ in \eqref{fangcheng12} are expressed as:
$|\widehat{y}_{3,0}^{\textbf{w}}|=\frac{1}{6}|\widehat{\textbf{z}}_{3}\widehat{\textbf{w}}_{0}|$,
$|\widehat{y}_{2,0}^{\textbf{w}}|=\frac{1}{6}|\widehat{\textbf{z}}_{2}\widehat{\textbf{w}}_{0}+\widehat{\textbf{z}}_3\widehat{\textbf{w}}_{1}|$,
$|\widehat{y}_{1,0}^{\textbf{w}}|=\frac{1}{6}|\widehat{\textbf{z}}_{1}\widehat{\textbf{w}}_{0}+\widehat{\textbf{z}}_{2}\widehat{\textbf{w}}_{1}+\widehat{\textbf{z}}_3\widehat{\textbf{w}}_{2}|$
and
$|\widehat{y}_{0,0}^{\textbf{w}}|=\frac{1}{6}|\widehat{\textbf{z}}_{0}\widehat{\textbf{w}}_{0}+\widehat{\textbf{z}}_1\widehat{\textbf{w}}_{1}+\widehat{\textbf{z}}_{2}\widehat{\textbf{w}}_{2}|$.
The terms $\widehat{\textbf{z}}_{k}\widehat{\textbf{w}}_{l}$  on which  $\widehat{y}_{4-n,0}^{\textbf{w}}$
is   dependent  are arranged  as follows,
\begin{align}\label{JHRLL}
\begin{matrix}
n=1&&&&&\widehat{\textbf{z}}_3\widehat{\textbf{w}}_0  \\
n=2&&&&\widehat{\textbf{z}}_2\widehat{\textbf{w}}_0&\widehat{\textbf{z}}_3\widehat{\textbf{w}}_1 \\
n=3&&&\widehat{\textbf{z}}_1\widehat{\textbf{w}}_0&\widehat{\textbf{z}}_2\widehat{\textbf{w}}_1&\widehat{\textbf{z}}_3\widehat{\textbf{w}}_2\\
n=4&&\widehat{\textbf{z}}_0\widehat{\textbf{w}}_0&\widehat{\textbf{z}}_{1}\widehat{\textbf{w}}_1&\widehat{\textbf{z}}_2\widehat{\textbf{w}}_2 \\
\end{matrix}
\end{align}
Based on  \eqref{JHRLL}  the terms $\widehat{\textbf{z}}_{k}\widehat{\textbf{w}}_{l}\omega^{lm}$ for  $\widehat{y}_{4-n,m}^{\textbf{w}}$
can be arranged similarly.
\end{exam}

 For the general case, similar to \eqref{JHRLL},
it follows from $2\le B\le \lceil\frac{N}{2}\rceil$, \eqref{tiaojian} and Proposition   \ref{jiegou}
that
the terms $\widehat{\textbf{z}}_{k}\widehat{\textbf{w}}_{l}$ on which  $|\widehat{y}^{\textbf{w}}_{\lfloor\frac{N}{2}\rfloor-(i+N-B)-n+1,0}|$ in \eqref{fangcheng12}
is   dependent are arranged  as follows,
% that the summation  on the right-hand side of  \eqref{fangcheng12} with   \textcolor[rgb]{0.44,0.00,0.94}{$m=0$} enjoys  the following structure:%in  \eqref{1stjiegou}.$n$th row
%\begin{figure}[htbp]\centering
\begin{small}\begin{align}\label{1stjiegou}
\begin{matrix}
&&&&&\widehat{\textbf{z}}_{\lfloor\!\frac{N}{2}\!\rfloor}\widehat{\textbf{w}}_{i\!+\!N\!-\!B}  \\
&&&&\widehat{\textbf{z}}_{\lfloor\!\frac{N}{2}\!\rfloor\!-\!1}\widehat{\textbf{w}}_{i\!+\!N\!-\!B}&\widehat{\textbf{z}}_{\lfloor\!\frac{N}{2}\!\rfloor}\widehat{\textbf{w}}_{i\!+\!N\!-\!B\!+\!1} \\
&&&\begin{rotate}{60}$\ddots$\end{rotate}&&\vdots\\
%0,\cdots,0,\widehat{z}_{\lfloor\frac{N}{2}\rfloor+2-B}\widehat{\textbf{w}}_{i+N-B},\widehat{z}_{\lfloor\frac{N}{2}\rfloor+3-B}\widehat{\textbf{w}}_{i+N-B+1},\cdots,\widehat{z}_{\lfloor\frac{N}{2}\rfloor}\widehat{\textbf{w}}_{i-2}  \nonumber\\
&&\widehat{\textbf{z}}_{\lfloor\!\frac{N}{2}\!\rfloor\!+\!1\!-\!B}\widehat{\textbf{w}}_{i\!+\!N\!-\!B}&\widehat{\textbf{z}}_{\lfloor\!\frac{N}{2}\!\rfloor\!+\!2\!-\!B}\widehat{\textbf{w}}_{i\!+\!N\!-\!B\!+\!1}&\cdots&\widehat{\textbf{z}}_{\lfloor\!\frac{N}{2}\!\rfloor}\widehat{\textbf{w}}_{i\!+\!N\!-\!1}\\
%&&\widehat{z}_{\lfloor\frac{N}{2}\rfloor-B}\widehat{\textbf{w}}_{i+N-B}&\widehat{z}_{\lfloor\frac{N}{2}\rfloor+1-B}\widehat{\textbf{w}}_{i+N-B+1}&\cdots&\widehat{z}_{\lfloor\frac{N}{2}\rfloor-1}\widehat{\textbf{w}}_{i-1}\\
&\begin{rotate}{60}$\ddots$\end{rotate}&&&\begin{rotate}{60}$\ddots$\end{rotate}\\
%&\widehat{z}_1\widehat{\textbf{w}}_{i+N-B}&\widehat{z}_2\widehat{w}_{i+N-B+1}&\cdots&\widehat{z}_{B}\widehat{w}_{i-1} \nonumber\\
\widehat{\textbf{z}}_0\widehat{\textbf{w}}_{i\!+\!N\!-\!B}&\widehat{\textbf{z}}_{1}\widehat{\textbf{w}}_{i\!+\!N\!-\!B\!+\!1}&\cdots&\widehat{\textbf{z}}_{B\!-\!1}\widehat{\textbf{w}}_{i\!+\!N\!-\!1} \\
\end{matrix}
\end{align}\end{small}
For $n=1$, as implied on the first row of \eqref{1stjiegou} the corresponding  measurement $|\widehat{y}^{\textbf{w}}_{\lfloor\frac{N}{2}\rfloor-(i+N-B),0}|$
is involved with only the term $\widehat{\textbf{z}}_{\lfloor\frac{N}{2}\rfloor}\widehat{\textbf{w}}_{i+N-B}$.
An observation on \eqref{fangcheng12} gives us that,
for $\big|\widehat{y}^{\textbf{w}}_{\lfloor\frac{N}{2}\rfloor-(i+N-B)-n+1,m}\big|$
the related terms  $\widehat{\textbf{z}}_{k}\widehat{\textbf{w}}_{l}\omega^{lm}$
are arranged as in \eqref{1stjiegou}. The following is on the determination of $\widehat{\textbf{z}}_{\lfloor\frac{N}{2}\rfloor}$.

\begin{lem}\label{theorem 1}
Suppose that the window $\textbf{w}\in \mathbb{C}^N$ is  $B$-bandlimited  such that
$2\le B\le \lceil\frac{N}{2}\rceil$. Consequently,  there exists
$i\in \{0, \ldots ,N-1\}$ such that $\widehat{\textbf{w}}_i=\cdots=\widehat{\textbf{w}}_{i+N-B-1}=0$
and $\widehat{\textbf{w}}_{i+N-B}\neq0.$
%Moreover, it is required that  $\widehat{\textbf{w}}_{i+N-B+1}\neq0.$
Then for any  analytic signal $\textbf{z}\in \mathbb{C}^N$ with DFT $\widehat{\textbf{z}}=(\widehat{\textbf{z}}_0,\widehat{\textbf{z}}_1,\ldots,\widehat{\textbf{z}}_{\lfloor\frac{N}{2}\rfloor},0,\ldots,0)$, we have
the following:

\textbf{Case I}: If  \textbf{$N$ is even},  then the component  $\widehat{\textbf{z}}_{\lfloor\frac{N}{2}\rfloor}$ can be determined (up to a sign) by the measurement $|\widehat{y}^{\textbf{w}}_{\lfloor\frac{N}{2}\rfloor-(i+N-B),0}|.$

\textbf{Case II}:  If  \textbf{$N$ is odd}, then the component $\widehat{\textbf{z}}_{\lfloor\frac{N}{2}\rfloor}$ can be determined (up to a unimodular scalar) by
$|\widehat{y}^{\textbf{w}}_{\lfloor\frac{N}{2}\rfloor-(i+N-B),0}|.$
\end{lem}
\begin{proof}
%???2???a$z$?????????????????|????????a??$w\in C^N$,???2????$\widehat{z}_0\in R,\widehat{\textbf{w}}_{\frac{N-1}{2}+i}\in C$????a??
It follows from \eqref{tiaojian} and \eqref{fangcheng12} that
$|\widehat{y}^{\textbf{w}}_{\lfloor\frac{N}{2}\rfloor-(i+N-B),0}|=\frac{1}{N}|\widehat{\textbf{z}}_{\lfloor\frac{N}{2}\rfloor}\widehat{\textbf{w}}_{i+N-B}|$.
Then $\widehat{\textbf{z}}_{\lfloor\frac{N}{2}\rfloor}=\frac{N|\widehat{y}^{\textbf{w}}_{\lfloor\frac{N}{2}\rfloor-(i+N-B),0}|}{|\widehat{\textbf{w}}_{i+N-B}|}e^{\textbf{i}\theta_0}$.
The proof for the odd case is completed.
By Proposition \ref{jiegou} for $N$ being even  we have $\widehat{\textbf{z}}_{\frac{N}{2}}\in \mathbb{R}$.
Then  $\widehat{\textbf{z}}_{\frac{N}{2}}=\epsilon\frac{N|\widehat{y}^{\textbf{w}}_{\frac{N}{2}-(i+N-B),0}|}{|\widehat{\textbf{w}}_{i+N-B}|}$ with $\epsilon\in\{1, -1\}$.
This completes the proof for the even case.
%For the case that $N$ being odd, it can be similarly proved.
%The proof is completed.
\end{proof}

Now it is ready to  establish the first main theorem.
\begin{theo}\label{abc}
Suppose that the window  $\textbf{w}\in \mathbb{C}^N$ is  $B$-bandlimited   such that $2\le B\le \lceil\frac{N}{2}\rceil$ and
there exists
$i\in \{0, \ldots ,N-1\}$ such that \eqref{tiaojian} holds and $\widehat{\textbf{w}}_{i+N-B+1}\neq0$.
%\begin{align}\label{J4FT} \widehat{\textbf{w}}_i=\cdots=\widehat{\textbf{w}}_{i+N-B-1}=0,
%\widehat{\textbf{w}}_{i+N-B+1}\neq0.\end{align}
Moreover,  we assume that the STFT separation parameter $0<L<N$ satisfies  $\lceil N/L\rceil\geq3$,
and choose any three distinct   numbers $m_{1}, m_{2}, m_{3}$ from   $\{0, 1, \ldots, \lceil N/L\rceil-1\}$.
Then   any  generic   analytic signal $\textbf{z}\in \mathbb{C}^N$ can be determined, up to a global sign, by  its $(3\lfloor\frac{N}{2}\rfloor+1)$ number of STFT measurements
\begin{align} \label{diyizhongleixingceliang} \big\{|\widehat{y}^{\textbf{w}}_{\lfloor\!\frac{N}{2}\!\rfloor-(i+N-B),0}|,|\widehat{y}^{\textbf{w}}_{k-(i+N-B),m_{j}}|: k=0,
\ldots, \lfloor\frac{N}{2}\rfloor-1,
j=1,2,3\big\}.\end{align}
%
%\begin{align} \label{YYY} \big\{|\widehat{y}^{\textbf{w}}_{\lfloor\!\frac{N}{2}\!\rfloor-(i\!+\!N\!-\!B),0}|,|\widehat{y}^{\textbf{w}}_{k-(i\!+\!N\!-\!B),m_{j}}|:k=\lfloor\!\frac{N}{2}\!\rfloor\!-\!1,\ldots,0,\nonumber\\
%\quad\quad\quad0\leqslant m_{j}\leqslant \lceil N/L\rceil-1,
%j=1,2,3\big\}.\end{align}
%
%
%
%\textbf{Case II}: \textbf{$N$ is odd}.  The signal  $\textbf{z}$ can be determined, up to a global phase, by the
%measurements in \eqref{YYY}.
\end{theo}

\begin{proof}  We mainly prove for the case when $N$ is even since the proof for the odd $N$ case is very similar. We will complete it by induction.
By Lemma  \ref{theorem 1}, the component $\widehat{\textbf{z}}_{\frac{N}{2}}$
can be determined up to  a  sign  by the measurement $|\widehat{y}^{\textbf{w}}_{\frac{N}{2}-(i+N-B),0}|.$
Denote such  a determination result by   $\epsilon\widehat{\textbf{z}}_{\frac{N}{2}}$ with $\epsilon\in \{1,-1\}$.
In what follows, we discuss how to recover other components $\widehat{\textbf{z}}_{0}, \ldots,
 \widehat{\textbf{z}}_{\frac{N}{2}-1}$.

We first address the recovery of  $\widehat{\textbf{z}}_{\frac{N}{2}-1}$ by the
%\textcolor[rgb]{0.00,1.00,0.00}{ Step 2 of Approach \ref{APPROACH11}}
STFT measurements $\{|\widehat{y}^{\textbf{w}}_{\frac{N}{2}-1-(i+N-B),m_j}|: j=1,2, 3\} $.
Consider the equation system
w.r.t $\widehat{\mathring{\textbf{z}}}_{\frac{N}{2}-1}$:
\begin{align}\label{222new1}
\begin{array}{ll}|\widehat{y}^{\textbf{w}}_{\frac{N}{2}-1-(i+N-B),m_j}|&=\frac{1}{N}\big|\widehat{\mathring{\textbf{z}}}_{\frac{N}{2}-1}\widehat{\textbf{w}}_{i+N-B}\omega^{(i+N-B)m_j}+\epsilon\widehat{\textbf{z}}_{\frac{N}{2}}\widehat{\textbf{w}}_{i+N-B+1}\omega^{(i+N-B+1)m_j}\big|,\\
&\quad\quad\quad\quad j=1,2,3.
\end{array}
\end{align}
Note that \eqref{222new1}
 is equivalent to
\begin{align}\label{222new}
\frac{N|\widehat{y}^{\textbf{w}}_{\frac{N}{2}-1-(i+N-B),m_j}|}{|\widehat{\textbf{w}}_{i+N-B}\omega^{(i+N-B)m_j}|} =\big|\widehat{\mathring{\textbf{z}}}_{\frac{N}{2}-1}+v_{j,\frac{N}{2}-1}\big|, j=1,2, 3,
\end{align}
%\textcolor[rgb]{0.00,1.00,0.00}{One can  check that \eqref{222} holds, namely
%\begin{align}\label{222new}
%\frac{N|\widehat{y}^{\textbf{w}}_{\frac{N}{2}-1-(i+N-B),m_j}|}{|\widehat{\textbf{w}}_{i+N-B}\omega^{(i+N-B)m_j}|} =\big|\widehat{\textbf{z}}_{\frac{N}{2}-1}+v_{j,\frac{N}{2}-1}\big|,
%\end{align}}
where
%\textcolor[rgb]{0.00,1.00,0.00}{$v_{j,\frac{N}{2}-1}$ is expressed in \eqref{2558}}
\begin{equation}\label{2558}
v_{j,\frac{N}{2}-1}:=\frac{\epsilon\widehat{\textbf{z}}_{\frac{N}{2}}\widehat{\textbf{w}}_{i+N-B+1}\omega^{(i+N-B+1)m_j}}{\widehat{\textbf{w}}_{i+N-B}\omega^{(i+N-B)m_j}}.
\end{equation}
For the generic analytic signal $\textbf{z}$, we have  $\widehat{\textbf{z}}_{\frac{N}{2}}\neq0$.
Therefore, for $v_{j,\frac{N}{2}-1}$ in \eqref{2558} we have
\begin{align}\label{KKKK234}\frac{v_{1,\frac{N}{2}-1}-v_{2,\frac{N}{2}-1}}{v_{1,\frac{N}{2}-1}-v_{3,\frac{N}{2}-1}}=\frac{\omega^{m_1}-\omega^{m_2}}{\omega^{m_1}-\omega^{m_3}}.\end{align}
By \eqref{KKKK234} and  Lemma \ref{m}  we have
$\Im\Big(\dfrac{v_{1,\frac{N}{2}-1}-v_{2,\frac{N}{2}-1}}{v_{1,\frac{N}{2}-1}-v_{3,\frac{N}{2}-1}}\Big)\ne0$.
Then  it follows from  Lemma \ref{L} that  there exists  a unique solution to the equation system \eqref{222new} w.r.t $\widehat{\mathring{\textbf{z}}}_{\frac{N}{2}-1}$.
Clearly, $\epsilon\widehat{\textbf{z}}_{\frac{N}{2}-1}$ is a solution.
Then it is the unique one. In what follows, we address how to    recover
the other components $\widehat{\textbf{z}}_{\frac{N}{2}-2}, \ldots, \widehat{\textbf{z}}_{0}$.
%\textcolor[rgb]{0.00,1.00,0.00}{As in Step 3 of Approach \ref{APPROACH11}, }
Suppose that $\epsilon\widehat{\textbf{z}}_{k}$ has been obtained  for any  $k\in \{\frac{N}{2}, \frac{N}{2}-1, \ldots, k_{0}\}$ with $k_{0}\in \{\frac{N}{2}, \frac{N}{2}-1, \ldots, 1\}$
 by   the measurements
%%% ?????????|?????????o2??2a?????????━|????a?
\begin{align}
\notag\big\{|\widehat{y}^{\textbf{w}}_{\frac{N}{2}-(i+N-B),0}|,|\widehat{y}^{\textbf{w}}_{\ell-(i+N-B),m_{j}}|: \ell=\frac{N}{2}-1,\ldots,k_{0},
j=1,2,3\big\}.
\end{align}
Now we discuss how to  recover $\widehat{\textbf{z}}_{k_0-1}$.
Consider the equation system w.r.t $\widehat{\mathring{\textbf{z}}}_{k_0-1}$:
\begin{equation}
\begin{array}{lll}
\label{886}|\widehat{y}^{\textbf{w}}_{k_0-1-(i+N-B),m_j}|
 &=\frac{1}{N}\big|\widehat{\mathring{\textbf{z}}}_{k_0-1}\widehat{\textbf{w}}_{i+N-B}\omega^{(i+N-B)m_j}\\ &\quad+ \sum\limits_{l=1}^{\frac{N}{2}-k_{0}+1}\epsilon\widehat{\textbf{z}}_{k_{0}-1+l}\widehat{\textbf{w}}_{i+N-B+l}\omega^{(i+N-B+l)m_j} \big|,
 j=1,2,3.
\end{array}
\end{equation}
Note that \eqref{886} is equivalent to
\begin{equation} \label{other01}
\frac{N|\widehat{y}^{\textbf{w}}_{k_0-1-(i+N-B),m_j}|}{|\widehat{\textbf{w}}_{i+N-B}\omega^{(i+N-B)m_j}|}=\big|\widehat{\mathring{\textbf{z}}}_{k_0-1}+v_{j,k_0-1} \big|, j=1,2, 3,\end{equation}
where
\begin{align}\label{45582}\begin{array}{lll}
 v_{j,k_{0}-1}:=\frac{\displaystyle \epsilon\widehat{\textbf{z}}_{k_{0}}\widehat{\textbf{w}}_{i+N-B+1}\omega^{(i+N-B+1)m_j}+\sum\limits_{l=2}^{\frac{N}{2}-k_{0}+1}\epsilon\widehat{\textbf{z}}_{k_{0}-1+l}\widehat{\textbf{w}}_{i+N-B+l}\omega^{(i+N-B+l)m_j}}{\widehat{\textbf{w}}_{i+N-B}\omega^{(i+N-B)m_j}}.
\end{array}
%v_{2,k}:=\frac{\displaystyle\sum\limits_{l=1}^{\frac{N}{2}-k}\epsilon\widehat{\textbf{z}}_{k+l}\widehat{\textbf{w}}_{i+N-B+l}\omega^{(i+N-B+l)m_2}}{\widehat{\textbf{w}}_{i+N-B}\omega^{(i+N-B)m_2}},\\
%v_{3,k}:=\frac{\displaystyle\sum\limits_{l=1}^{\frac{N}{2}-k}\epsilon\widehat{\textbf{z}}_{k+l}\widehat{\textbf{w}}_{i+N-B+l}\omega^{(i+N-B+l)m_3}}{\widehat{\textbf{w}}_{i+N-B}\omega^{(i+N-B)m_3}}.
\end{align}
%
%
%\textbf{Recovery  of} $\widehat{\textbf{z}}_{\frac{N}{2}-2}$ by  step 3 of Approach \ref{APPROACH21}.  It is easy to check that
%\begin{align}\label{455}
%\frac{N|\widehat{y}^{\textbf{w}}_{\frac{N}{2}-2-(i+N-B),m_j}|}{|\widehat{\textbf{w}}_{i+N-B}\omega^{(i+N-B)m_j}|}=\big|\widehat{\textbf{z}}_{\frac{N}{2}-2}+v_{j,\frac{N}{2}-2}\big|, j=1,2,3,
%\end{align}
%where%\Im(f(z))&=\frac{1}{|cz+d|^2}\{\Im(a\bar{c})(x^2+y^2)+(\Im(a\bar{d})\\&+\Im(b\bar{c}))x+(\Re(a\bar{d})-\Re(b\bar{c}))y+\Im(b\bar{d})\}.
%\begin{equation}\label{3558}\begin{array}{lll}
%v_{j,\frac{N}{2}-2}\\
%=\frac{1}{\widehat{\textbf{w}}_{i+N-B}\omega^{(i+N-B)m_j}}\Big\{\epsilon\widehat{\textbf{z}}_{\frac{N}{2}-1}\widehat{\textbf{w}}_{i+N-B+1}\\
%\omega^{(i+N-B+1)m_j}+\epsilon\widehat{\textbf{z}}_{\frac{N}{2}}\widehat{\textbf{w}}_{i+N-B+2}\omega^{(i+N-B+2)m_j}\Big\}.
%\end{array}
%\end{equation}
% Now we consider the above equation system w.r.t
%$\widehat{\textbf{z}}_{\frac{N}{2}-2}$.
Motivated by Lemma \ref{L},
define
\begin{align}
\notag f(\widehat{\textbf{z}}_{k_{0}})&:=\frac{v_{1,k_{0}-1}-v_{2,k_{0}-1}}{v_{1,k_{0}-1}-v_{3,k_{0}-1}}\\
&=\dfrac{a\widehat{\textbf{z}}_{k_{0}}+b}{c\widehat{\textbf{z}}_{k_{0}}+d},
\end{align}
where
\begin{align}\label{up1133}
\left\{\begin{array}{lll}
a=\epsilon\widehat{\textbf{w}}_{i+N-B+1}(\omega^{m_1}-\omega^{m_2}),\\
b=\sum\limits_{l=2}^{\frac{N}{2}-k_{0}+1}\epsilon\widehat{\textbf{z}}_{k_{0}-1+l}\widehat{\textbf{w}}_{i+N-B+l}(\omega^{lm_1}-\omega^{lm_2}),\\
c=\epsilon\widehat{\textbf{w}}_{i+N-B+1}(\omega^{m_1}-\omega^{m_3}),\\
d=\sum\limits_{l=2}^{\frac{N}{2}-k_{0}+1}\epsilon\widehat{\textbf{z}}_{k_{0}-1+l}\widehat{\textbf{w}}_{i+N-B+l}(\omega^{lm_1}-\omega^{lm_3}).
\end{array}\right.
\end{align}
Recall that $\widehat{\textbf{w}}_{_{i+N-B+1}}\neq0$
and $m_{1}, m_{2}, m_{3}$ are distinct.
Then  $ac\neq0$. For the generic analytic signal $\textbf{z}$, we have $ad-bc\neq0.$
That is,    $f(\widehat{\textbf{z}}_{k_{0}})$
meets  the requirements  in  Lemma \ref{lemma3.2}.
Then  $\Im[f(\widehat{\textbf{z}}_{k_{0}})]\neq0$.
Therefore, by Lemma \ref{L} the component  $\epsilon\widehat{\textbf{z}}_{k_{0}-1}$
 can be determined by the equation system \eqref{other01}. Through  the induction procedures, the proof can be completed.

 For $N$ being odd, as in the even case the recovery  starts with $\widehat{\textbf{z}}_{\frac{N-1}{2}}$.
Suppose that what we get is
$\widehat{\textbf{z}}_{\frac{N-1}{2}}e^{\textbf{i}\widehat{\theta}_{0}}$.
 Through the similar recursive procedures as in \eqref{other01}, what we get is
 $e^{\textbf{i}\widehat{\theta}_{0}}(\widehat{\textbf{z}}_{0}, \ldots,
 \widehat{\textbf{z}}_{\frac{N-1}{2}})$.
 Recall that  $\widehat{\textbf{z}}_{0}$
is real. Then one needs to choose a phase
$\tilde{\theta}$ such that $e^{\textbf{i}\tilde{\theta}}e^{\textbf{i}\widehat{\theta}_{0}}\widehat{\textbf{z}}_{0}$ is real.
That is, what we get is $\epsilon
(\widehat{\textbf{z}}_{0}, \ldots,
 \widehat{\textbf{z}}_{\frac{N-1}{2}})$
 with $\epsilon\in\{1, -1\}$.
 This completes the proof.
\end{proof}

\begin{rem}\label{diyigeyaoqiu}
In Theorem \ref{abc} it is required that
$\widehat{\textbf{w}}_{i+N-B+1}\neq0$. Such a requirement  is crucial for
the determination of $\widehat{\textbf{z}}_{\lfloor\frac{N}{2}\rfloor-1}$.
 If it is  not satisfied, then
the equation system w.r.t
$\widehat{\mathring{\textbf{z}}}_{\lfloor\frac{N}{2}\rfloor-1}$:
$$
\frac{N|\widehat{y}^{\textbf{w}}_{\lfloor\frac{N}{2}\rfloor-1-(i+N-B),m_j}|}{|\widehat{\textbf{w}}_{i+N-B}\omega^{(i+N-B)m_j}|} =\big|\widehat{\mathring{\textbf{z}}}_{\lfloor\frac{N}{2}\rfloor-1}+\frac{\widehat{\textbf{z}}_{\lfloor\frac{N}{2}\rfloor}\widehat{\textbf{w}}_{i+N-B+1}\omega^{(i+N-B+1)m_j}}{\widehat{\textbf{w}}_{i+N-B}\omega^{(i+N-B)m_j}}\big|,j=1,2,3
$$
%\begin{equation}\label{22}
%\frac{N|\widehat{y}^{\textbf{w}}_{\frac{N}{2}-1-(i+N-B),m_j}|}{|\wi%dehat{\textbf{w}}_{i+N-B}\omega^{(i+N-B)m_j}|} %=\big|\widehat{\mathring{\textbf{z}}}_{\frac{N}{2}-1}+v_{j,\frac{N}{2}-1}\big|%, j=1,2,3,
%\end{equation}
degenerates to
$$\frac{N|\widehat{y}^{\textbf{w}}_{\lfloor\frac{N}{2}\rfloor-1-(i+N-B),m_j}|}{|\widehat{\textbf{w}}_{i+N-B}\omega^{(i+N-B)m_j}|} =\big|\widehat{\mathring{\textbf{z}}}_{\lfloor\frac{N}{2}\rfloor-1}\big|, j=1,2,3.
$$
Clearly, the above  system  is underdetermined and $\widehat{\textbf{z}}_{\lfloor\frac{N}{2}\rfloor-1}$ can not be determined.
\end{rem}

\subsection{The second main result:    window bandlimit   $B=\lceil\frac{N}{2}\rceil+1$ case}\label{secondmainresult}
Suppose that   the window   $\textbf{w}\in \mathbb{C}^N$ is $(\lceil\frac{N}{2}\rceil+1)$-bandlimited. Consequently,  there exists
$i\in \{0, \ldots ,N-1\}$ such that \begin{align}\label{tiaojian1} \widehat{\textbf{w}}_i=\cdots=\widehat{\textbf{w}}_{i+\lfloor\frac{N}{2}\rfloor-2}=0,
\widehat{\textbf{w}}_{i+\lfloor\frac{N}{2}\rfloor-1}\neq0.\end{align}
We are interested in   the    STFT measurements  at $(2-i+N-n, m)$:
\begin{align}\label{fangcheng121} \big|\widehat{y}^{\textbf{w}}_{2-i+N-n,m}\big|=\dfrac{1}{N}\big|\sum_{l=0}^{N-1}\widehat{\textbf{z}}_{2-i+N-n+l}\widehat{\textbf{w}}_l\omega^{lm}\big|,  \end{align}
where  $n=1, \ldots,\lfloor\frac{N}{2}\rfloor$.
%where $k_{n}=\lfloor\frac{N}{2}\rfloor-(i+N-B)-n+1.$
%$\lfloor\frac{N}{2}\rfloor-(i+N-B)=k_1>k_2>\cdots>k_{\lfloor\frac{N}{2}\rfloor+1}=N-i+B$

 Again the following is a motivation example for the structure of the summation in \eqref{fangcheng121}.

\begin{exam}
Let $N=6$ and the window bandlimit  $B=\lceil\frac{N}{2}\rceil+1=4.$ For an analytic signal $\textbf{z}\in \mathbb{C}^{6}$,
it follows from Proposition   \ref{jiegou} (i) that its
DFT $\widehat{\textbf{z}}=(\widehat{\textbf{z}}_0,\widehat{\textbf{z}}_1,\widehat{\textbf{z}}_2,\widehat{\textbf{z}}_3, 0,  0)$.
Choose a $4$-bandlimited window $\textbf{w}\in \mathbb{C}^{6}$ such that
$\widehat{\textbf{w}}=(\widehat{\textbf{w}}_0,\widehat{\textbf{w}}_1,\widehat{\textbf{w}}_2,\widehat{\textbf{w}}_3,0,0)$
and correspondingly  $i=4$ in \eqref{tiaojian1}.
For $n=1, 2, 3$, $|\widehat{y}_{4-n,0}^{\textbf{w}}|$ in \eqref{fangcheng121} are expressed as:
$|\widehat{y}_{3,0}^{\textbf{w}}|=\frac{1}{6}|\widehat{\textbf{z}}_{3}\widehat{\textbf{w}}_{0}+\widehat{\textbf{z}}_{0}\widehat{\textbf{w}}_{3}|$,
$|\widehat{y}_{2,0}^{\textbf{w}}|=\frac{1}{6}|\widehat{\textbf{z}}_{2}\widehat{\textbf{w}}_{0}+\widehat{\textbf{z}}_3\widehat{\textbf{w}}_{1}|$,
$|\widehat{y}_{1,0}^{\textbf{w}}|=\frac{1}{6}|\widehat{\textbf{z}}_{1}\widehat{\textbf{w}}_{0}+\widehat{\textbf{z}}_{2}\widehat{\textbf{w}}_{1}+\widehat{\textbf{z}}_3\widehat{\textbf{w}}_{2}|$.
The terms $\widehat{\textbf{z}}_{k}\widehat{\textbf{w}}_{l}$ on which  $\widehat{y}_{4-n,0}^{\textbf{w}}$
is   dependent are arranged  as follows,
\begin{align}\label{JHRLLl}
\begin{matrix}
n=1&\widehat{\textbf{z}}_0\widehat{\textbf{w}}_3&&&&\widehat{\textbf{z}}_3\widehat{\textbf{w}}_0  \\
n=2&&&&\widehat{\textbf{z}}_2\widehat{\textbf{w}}_0&\widehat{\textbf{z}}_3\widehat{\textbf{w}}_1 \\
n=3&&&\widehat{\textbf{z}}_1\widehat{\textbf{w}}_0&\widehat{\textbf{z}}_2\widehat{\textbf{w}}_1&\widehat{\textbf{z}}_3\widehat{\textbf{w}}_2\\
\end{matrix}
\end{align}
Based on  \eqref{JHRLLl}  the terms $\widehat{\textbf{z}}_{k}\widehat{\textbf{w}}_{l}\omega^{lm}$ of  $\widehat{y}_{4-n,m}^{\textbf{w}}$
can be arranged similarly.
\end{exam}

For the general case when the window bandlimit $B=\lceil\frac{N}{2}\rceil+1$, as in \eqref{JHRLLl},
it follows from  \eqref{tiaojian1} and Proposition   \ref{jiegou}
that
the terms $\widehat{\textbf{z}}_{k}\widehat{\textbf{w}}_{l}$ on which  $|\widehat{y}^{\textbf{w}}_{2-i+N-n,0}|$ in \eqref{fangcheng121}
is   dependent are arranged  as follows,

%%$\lfloor\frac{N}{2}\rfloor-(i+N-B)=k_1>k_2>\cdots>k_{\lfloor\frac{N}{2}\rfloor+1}=N-i+B$
%It  follows from \eqref{tiaojian1}, $B=\lceil\frac{N}{2}\rceil+1$ and Proposition \ref{jiegou} that
%the summation    on the right-hand side of  \eqref{fangcheng121} enjoys the  following structure:
%\eqref{1stjiegou1}.
%\begin{figure*}[htbp]\centering
\begin{align}\label{1stjiegou1}
\begin{matrix}
\widehat{\textbf{z}}_{0}\widehat{\textbf{w}}_{i-1+N}&&&&\widehat{\textbf{z}}_{\lfloor\frac{N}{2}\rfloor}\widehat{\textbf{w}}_{i+\lfloor\frac{N}{2}\rfloor-1}  \\
&&&\widehat{\textbf{z}}_{\lfloor\frac{N}{2}\rfloor-1}\widehat{\textbf{w}}_{i+\lfloor\frac{N}{2}\rfloor-1}&\widehat{\textbf{z}}_{\lfloor\frac{N}{2}\rfloor}\widehat{\textbf{w}}_{i+\lfloor\frac{N}{2}\rfloor} \\
&&\begin{rotate}{60}$\ddots$\end{rotate}&&\vdots\\
&\widehat{\textbf{z}}_{1}\widehat{\textbf{w}}_{i+\lfloor\frac{N}{2}\rfloor-1}&\widehat{\textbf{z}}_{2}\widehat{\textbf{w}}_{i+\lfloor\frac{N}{2}\rfloor}&\cdots&\widehat{\textbf{z}}_{\lfloor\frac{N}{2}\rfloor}\widehat{\textbf{w}}_{i+2\lfloor\frac{N}{2}\rfloor-2} \\
\end{matrix}
\end{align}%\end{figure*}
An observation on \eqref{fangcheng121} gives us that,
for $\big|\widehat{y}^{\textbf{w}}_{2-i+N-n,m}\big|$
the related terms  $\widehat{\textbf{z}}_{k}\widehat{\textbf{w}}_{l}\omega^{lm}$
are arranged as in \eqref{1stjiegou1}.
Motivated by such a structure, we next use the (multi-window)  measurements $\big\{ |\widehat{y}_{2-i+N-n,m}^{\textbf{w}(s)}|:s=1, 2, 3, 4 \big\} $  to do the PR for $\textbf{z}$, which is stated below  as our second main theorem.

\begin{theo}\label{999}  Assume that  the STFT separation parameter $L$ satisfies  $\lceil N/L\rceil\geq3$.
Suppose that the four windows  $\textbf{w}^{(s)}\in \mathbb{C}^{N}, s=1, \ldots, 4$
 are   $(\lceil\frac{N}{2}\rceil+1)$-bandlimited such that they
 satisfy  \eqref{tiaojian1} with  $i\in \{0, \ldots ,N-1\}$,
 $\widehat{\textbf{w}}_{i+\lfloor\frac{N}{2}\rfloor}^{(1)}\neq0$, and let
$m_{1}, m_{2}, m_{3}\in \{0, 1, \ldots, \lceil N/L\rceil-1\}$ be three distinct   numbers.
 If
the  matrix
\begin{align}\label{JJK}\mathcal{A}_{0}:=\left(\begin{array}{llllllll}
a_{11}^{(1)}&a_{12}^{(1)}&a_{21}^{(1)}&a_{22}^{(1)}\\
a_{11}^{(2)}&a_{12}^{(2)}&a_{21}^{(2)}&a_{22}^{(2)}\\
a_{11}^{(3)}&a_{12}^{(3)}&a_{21}^{(3)}&a_{22}^{(3)}\\
a_{11}^{(4)}&a_{12}^{(4)}&a_{21}^{(4)}&a_{22}^{(4)}\end{array}\right)\end{align}
is invertible,
where
\begin{equation}\label{AAA}
\left\{\begin{aligned}
a_{11}^{(s)}=|\widehat{\textbf{w}}^{(s)}_{i+\lfloor\frac{N}{2}\rfloor-1}|^2,\ a_{12}^{(s)}=\widehat{\textbf{w}}^{(s)}_{i+\lfloor\frac{N}{2}\rfloor-1}\overline{\widehat{\textbf{w}}^{(s)}_{i-1+N}},\\a_{21}^{(s)}=\widehat{\textbf{w}}^{(s)}_{i-1+N}\overline{\widehat{\textbf{w}}^{(s)}_{i+\lfloor\frac{N}{2}\rfloor-1}},\ a_{22}^{(s)}=|\widehat{\textbf{w}}^{(s)}_{i-1+N}|^2,\end{aligned}\right.
\end{equation}
then % \textcolor[rgb]{0.00,1.00,0.00}{ through Approach \ref{APPROACH21},  }
any generic analytic signal $\textbf{z}\in \mathbb{C}^N$ can be determined (up to a global sign) by its  $(3\lfloor\frac{N}{2}\rfloor+1)$ number of STFT measurements %\textcolor[rgb]{0.00,1.00,0.00}{in \eqref{4wmeasuements}}
\begin{align}\label{4wmeasuements}  \begin{array}{lll}\Big\{|\widehat{y}_{1-i+N,0}^{\textbf{w}(1)}|,|\widehat{y}_{1-i+N,0}^{\textbf{w}(2)}|, |\widehat{y}_{1-i+N,0}^{\textbf{w}(3)}|, |\widehat{y}_{1-i+N,0}^{\textbf{w}(4)}|,
|\widehat{y}_{k-(i+\lfloor\frac{N}{2}\rfloor-1),m_{j}}^{\textbf{w}(1)}|: k=1, \ldots,\lfloor\frac{N}{2}\rfloor-1,\\
\quad\quad\quad\quad\quad\quad  j=1,2,3\Big\}.\end{array}\end{align}
%\begin{align}\label{YHD} \begin{array}{lll}\Big\{|\widehat{y}_{1-i+N,0}^{\textbf{w}(1)}|,|\widehat{y}_{1-i+N,0}^{\textbf{w}(2)}|, |\widehat{y}_{1-i+N,0}^{\textbf{w}(3)}|, |\widehat{y}_{1-i+N,0}^{\textbf{w}(4)}|,\\
%\quad\quad\quad|\widehat{y}_{k-(i+\lfloor\frac{N}{2}\rfloor-1),m_{j}}^{\textbf{w}(1)}|: k=\lfloor\frac{N}{2}\rfloor-1,\ldots,1,\\
%\quad\quad\quad\quad\quad\quad 0\leqslant m_{j}\leqslant \lceil N/L\rceil-1,j=1,2,3\Big\}.\end{array}\end{align}
 %?????2?????????????????????????????????━??━???o????|??????????????????????????|?????????????━??━??????????a
\end{theo}

\begin{proof}
Consider the equation system
w.r.t $(\widehat{\mathring{\textbf{z}}}_0, \widehat{\mathring{\textbf{z}}}_{\lfloor\frac{N}{2}\rfloor})$:
\begin{equation}\label{666}|\widehat{y}^{\textbf{w}(s)}_{1-i+N,0}|=\frac{1}{N}|\widehat{\mathring{\textbf{z}}}_{\lfloor\frac{N}{2}\rfloor}\widehat{\textbf{w}}^{(s)}_{i+\lfloor\frac{N}{2}\rfloor-1}+\widehat{\mathring{\textbf{z}}}_0\widehat{\textbf{w}}^{(s)}_{i-1+N}|, s=1,2,3,4.\end{equation}
Note that \eqref{666}
is equivalent to
\begin{align} \label{K4GH} \big(\widehat{\mathring{\textbf{z}}}_{\lfloor\frac{N}{2}\rfloor}\widehat{\textbf{w}}_{i+\lfloor\frac{N}{2}\rfloor-1}^{(s)}\!+\widehat{\mathring{\textbf{z}}}_0\widehat{\textbf{w}}_{i-1+N}^{(s)}\big)\overline{\big(\widehat{\mathring{\textbf{z}}}_{\lfloor\frac{N}{2}\rfloor}\widehat{\textbf{w}}_{i+\lfloor\frac{N}{2}\rfloor-1}^{(s)}\!+\widehat{\mathring{\textbf{z}}}_0\widehat{\textbf{w}}_{i-1+N}^{(s)}\big)}=
N^2|\widehat{y}_{1-i+N,0}^{\textbf{w}(s)}|^2, s=1,2,3,4. \end{align}
Through the direct calculation,
\eqref{K4GH} is equivalent to
\begin{align}\label{huifuguochengxxxx1}
\begin{array}{llllllll}
\mathcal{A}_{0}\left(\begin{array}{llllllll}
 %\vspace{1.5ex}
|\widehat{\mathring{\textbf{z}}}_{\lfloor\frac{N}{2}\rfloor}|^2\\
%\vspace{1.5ex}
\widehat{\mathring{\textbf{z}}}_{\lfloor\frac{N}{2}\rfloor}\widehat{\mathring{\textbf{z}}}_{0} \\
%\vspace{1.5ex}
\overline{\widehat{\mathring{\textbf{z}}}_{\lfloor\frac{N}{2}\rfloor}}\widehat{\mathring{\textbf{z}}}_{0}\\
%\vspace{1.5ex}
\widehat{\mathring{\textbf{z}}}_{0}^2
\end{array}\right)=\left(\begin{array}{lllllllll}
%\vspace{1.5ex}
N^2|\widehat{y}_{1-i+N,0}^{\textbf{w}(1)}|^2\\
%\vspace{1.5ex}
N^2|\widehat{y}_{1-i+N,0}^{\textbf{w}(2)}|^2 \\
%\vspace{1.5ex}
N^2|\widehat{y}_{1-i+N,0}^{\textbf{w}(3)}|^2 \\
%\vspace{1.5ex}
N^2|\widehat{y}_{1-i+N,0}^{\textbf{w}(4)}|^2
\end{array}\right).
\end{array}
\end{align}
%\textcolor[rgb]{0.00,1.00,0.00}{where $\mathcal{A}_{0}$ is defined in Approach \ref{APPROACH21}  \eqref{JJK}.}
Since $\mathcal{A}_{0}$ is invertible and $\widehat{\textbf{z}}_0\in \mathbb{R}$,
$(\widehat{\textbf{z}}_0, \widehat{\textbf{z}}_{\lfloor\frac{N}{2}\rfloor})$ can be determined
up to a sign by the four  measurements in \eqref{666}. We denote such a recovery result by
$\epsilon(\widehat{\textbf{z}}_0, \widehat{\textbf{z}}_{\lfloor\frac{N}{2}\rfloor})$ with $\epsilon\in \{1,-1\}$.
In what follows, we discuss how to determine other components $\widehat{\textbf{z}}_1,\ldots, \widehat{\textbf{z}}_{\lfloor\frac{N}{2}\rfloor-1} $ of $\widehat{\textbf{z}}$.

We first address the recovery of  $\widehat{\textbf{z}}_{\lfloor\frac{N}{2}\rfloor-1}$ by  STFT magnitudes  $\{|\widehat{y}_{N-i,m_j}^{\textbf{w}(1)}|:j=1,2,3\}$.
Consider the equation system w.r.t $\widehat{\mathring{\textbf{z}}}_{\lfloor\frac{N}{2}\rfloor-1}$:  \begin{equation}\begin{array}{lll}\label{21299}
|\widehat{y}_{N-i,m_j}^{\textbf{w}(1)}| =\frac{1}{N}\big|\widehat{\mathring{\textbf{z}}}_{\lfloor\frac{N}{2}\rfloor-1}\widehat{\textbf{w}}_{i+\lfloor\frac{N}{2}\rfloor-1}^{(1)}\omega^{(i+\lfloor\frac{N}{2}\rfloor-1)m_j}+\epsilon\widehat{\textbf{z}}_{\lfloor\frac{N}{2}\rfloor}\widehat{\textbf{w}}_{i+\lfloor\frac{N}{2}\rfloor}^{(1)}\omega^{(i+\lfloor\frac{N}{2}\rfloor)m_j}\big|,j=1,2,3.
\end{array}
\end{equation}
Note that \eqref{21299} is equivalent to
%\textcolor[rgb]{0.00,1.00,0.00}{ Step 2 of Approach \ref{APPROACH21}}. One can  check that
\begin{equation}\label{21229}
%\left\{
%\begin{aligned}
\frac{N|\widehat{y}_{N-i,m_j}^{\textbf{w}(1)}|}{|\widehat{\textbf{w}}_{i+\lfloor\frac{N}{2}\rfloor-1}^{(1)}\omega^{(i+\lfloor\frac{N}{2}\rfloor-1)m_j}|} =\big|\widehat{\mathring{\textbf{z}}}_{\lfloor\frac{N}{2}\rfloor-1}+v_{j,\lfloor\frac{N}{2}\rfloor-1}\big|,j=1,2,3,
%\frac{N|\widehat{y}_{N-i,m_2}^{\textbf{w}(1)}|}{|\widehat{\textbf{w}}_{i+\lfloor\frac{N}{2}\rfloor-1}^{(1)}\omega^{(i+\lfloor\frac{N}{2}\rfloor-1)m_2}|} &=\big|\widehat{\textbf{z}}_{\lfloor\frac{N}{2}\rfloor-1}+v_{2,\lfloor\frac{N}{2}\rfloor-1}\big|,\\
%\frac{N|\widehat{y}_{N-i,m_3}^{\textbf{w}(1)}|}{|\widehat{\textbf{w}}_{i+\lfloor\frac{N}{2}\rfloor-1}^{(1)}\omega^{(i+\lfloor\frac{N}{2}\rfloor-1)m_3}|} &=\big|\widehat{\textbf{z}}_{\lfloor\frac{N}{2}\rfloor-1}+v_{3,\lfloor\frac{N}{2}\rfloor-1}\big|,
%\end{aligned}
%\right.
\end{equation}
where
$
%\left\{
%\begin{aligned}
v_{j,\lfloor\frac{N}{2}\rfloor-1}:=\frac{\epsilon\widehat{\textbf{z}}_{\lfloor\frac{N}{2}\rfloor}\widehat{\textbf{w}}_{i+\lfloor\frac{N}{2}\rfloor}^{(1)}\omega^{(i+\lfloor\frac{N}{2}\rfloor)m_j}}{\widehat{\textbf{w}}_{i+\lfloor\frac{N}{2}\rfloor-1}^{(1)}\omega^{(i+\lfloor\frac{N}{2}\rfloor-1)m_j}}.
%v_{2,\lfloor\frac{N}{2}\rfloor-1}:=\frac{\epsilon\widehat{\textbf{z}}_{\lfloor\frac{N}{2}\rfloor}\widehat{\textbf{w}}_{i+\lfloor\frac{N}{2}\rfloor}^{(1)}\omega^{(i+\lfloor\frac{N}{2}\rfloor)m_2}}{\widehat{\textbf{w}}_{i+\lfloor\frac{N}{2}\rfloor-1}^{(1)}\omega^{(i+\lfloor\frac{N}{2}\rfloor-1)m_2}},\\
%v_{3,\lfloor\frac{N}{2}\rfloor-1}:=\frac{\epsilon\widehat{\textbf{z}}_{\lfloor\frac{N}{2}\rfloor}\widehat{\textbf{w}}_{i+\lfloor\frac{N}{2}\rfloor}^{(1)}\omega^{(i+\lfloor\frac{N}{2}\rfloor)m_3}}{\widehat{\textbf{w}}_{i+\lfloor\frac{N}{2}\rfloor-1}^{(1)}\omega^{(i+\lfloor\frac{N}{2}\rfloor-1)m_3}}.
%\end{aligned}
%\right.
$
For the generic analytic signal $\textbf{z}\in \mathbb{C}^N$, we have  $\widehat{\textbf{z}}_{\lfloor\frac{N}{2}\rfloor}\neq0$.
Therefore, for $j=1,2,3$ we have
\begin{align}\label{JHLLL}\dfrac{v_{1,\lfloor\frac{N}{2}\rfloor-1}-v_{2,\lfloor\frac{N}{2}\rfloor-1}}{v_{1,\lfloor\frac{N}{2}\rfloor-1}-v_{3,\lfloor\frac{N}{2}\rfloor-1}}
=\frac{\omega^{m_1}-\omega^{m_2}}{\omega^{m_1}-\omega^{m_3}}.
\end{align}
By \eqref{JHLLL} and  Lemma \ref{m}  we have
$\Im\Big(\dfrac{v_{1,\lfloor\frac{N}{2}\rfloor-1}-v_{2,\lfloor\frac{N}{2}\rfloor-1}}{v_{1,\lfloor\frac{N}{2}\rfloor-1}-v_{3,\lfloor\frac{N}{2}\rfloor-1}}\Big)\ne0.$
Now it follows from  Lemma \ref{L} that  there exists  a unique solution to the equation system \eqref{21229} w.r.t $\widehat{\mathring{\textbf{z}}}_{\lfloor\frac{N}{2}\rfloor-1}$.
Clearly, $\epsilon\widehat{\textbf{z}}_{\lfloor\frac{N}{2}\rfloor-1}$ is a solution.
Then it is the unique solution. In what follows, we address how to recover
the other components $\widehat{\textbf{z}}_{\lfloor\frac{N}{2}\rfloor-2}, \ldots, \widehat{\textbf{z}}_{1}$.
%$\widehat{\textbf{z}}_{k}$,$k\leq \lfloor\frac{N}{2}\rfloor-2$.
Suppose that for any  $k\in \{\lfloor\frac{N}{2}\rfloor,
\lfloor\frac{N}{2}\rfloor-1, \ldots, k_{0},0\}$  where  $k_{0}\in \{\lfloor\frac{N}{2}\rfloor, \lfloor\frac{N}{2}\rfloor-1, \ldots, 2\}$,
the component  $\widehat{\textbf{z}}_{k}$ has been determined by  the measurements
\begin{align}\begin{array}{lll}
\notag\big\{|\widehat{y}^{\textbf{w}(s)}_{1-i+N,0}|,|\widehat{y}^{\textbf{w}(1)}_{\ell-(i+\lfloor\frac{N}{2}\rfloor-1),m_{j}}|: \ell=\lfloor\frac{N}{2}\rfloor-1,\ldots,k_{0},
s=1,2,3,4,
j=1,2,3\big\}.
\end{array}
\end{align}
%Suppose that $\epsilon\widehat{\textbf{z}}_{\lfloor\frac{N}{2}\rfloor},\ldots,\epsilon\widehat{\textbf{z}}_{k+1},\epsilon\widehat{\textbf{z}}_{0} $ have  been recovered where $k\leq
%\lfloor\frac{N}{2}\rfloor-2$.
We next recover  $\widehat{\textbf{z}}_{k_{0}-1}$.
Consider the equation system w.r.t $\widehat{\mathring{\textbf{z}}}_{k_{0}-1}$:  \begin{align}
\label{7865}\begin{array}{lllll}
\displaystyle |\widehat{y}_{k_{0}-1-(i+\lfloor\frac{N}{2}\rfloor-1),m_j}^{\textbf{w}(1)}|=&
\frac{1}{N}\big|\widehat{\mathring{\textbf{z}}}_{k_{0}-1}\widehat{\textbf{w}}_{i+\lfloor\frac{N}{2}\rfloor-1}^{(1)}\omega^{(i+\lfloor\frac{N}{2}\rfloor-1)m_j}\\
&+\displaystyle\sum\limits_{l=1}^{\lfloor\frac{N}{2}\rfloor-k_{0}+1}\epsilon\widehat{\textbf{z}}_{k_{0}-1+l}\widehat{\textbf{w}}_{i+\lfloor\frac{N}{2}\rfloor-1+l}^{(1)}\omega^{(i+\lfloor\frac{N}{2}\rfloor-1+l)m_j}\big|,j=1,2,3.
\end{array}
\end{align}
Note that  \eqref{7865} is equivalent to
%%Suppose that $\widehat{z}_{\lfloor\frac{N}{2}\rfloor},\widehat{z}_0$has been determined.
%for $k\le \lfloor\frac{N}{2}\rfloor-2$,
%,????$\{|\widehat{y}_{k-(i+\lfloor\frac{N}{2}\rfloor-1),m_j}^{(\textcolor[rgb]{0.00,0.07,1.00}{1})}|: 0\leqslant m_{j}\leqslant \lceil N/L\rceil-1, j=1,2,3\}$ ?2?????????,??|??????|????
\begin{equation}\label{42}
%\left\{
%\begin{aligned}
\frac{N|\widehat{y}_{k_{0}-1-(i+\lfloor\frac{N}{2}\rfloor-1),m_j}^{\textbf{w}(1)}|}{|\widehat{\textbf{w}}_{i+\lfloor\frac{N}{2}\rfloor-1}^{(1)}\omega^{(i+\lfloor\frac{N}{2}\rfloor-1)m_j}|} =\big|\widehat{\mathring{\textbf{z}}}_{k_{0}-1}+v_{j,k_{0}-1}\big|,j=1,2,3,
%\frac{N|\widehat{y}_{k-(i+\lfloor\frac{N}{2}\rfloor-1),m_2}^{\textbf{w}(1)}|}{|\widehat{\textbf{w}}_{i+\lfloor\frac{N}{2}\rfloor-1}^{(1)}\omega^{(i+\lfloor\frac{N}{2}\rfloor-1)m_2}|} &=\big|\widehat{\textbf{z}}_{k}+v_{2,k}\big|,\\
%\frac{N|\widehat{y}_{k-(i+\lfloor\frac{N}{2}\rfloor-1),m_3}^{\textbf{w}(1)}|}{|\widehat{\textbf{w}}_{i+\lfloor\frac{N}{2}\rfloor-1}^{(1)}\omega^{(i+\lfloor\frac{N}{2}\rfloor-1)m_3}|} &=\big|\widehat{\textbf{z}}_{k}+v_{3,k}\big|,
%\end{aligned}
%\right.
\end{equation}
where
\begin{equation}\label{658}
%\left\{
%\begin{aligned}
v_{j,k_{0}-1}:=\frac{\epsilon\widehat{\textbf{z}}_{k_{0}}\widehat{\textbf{w}}_{i+\lfloor\frac{N}{2}\rfloor}^{(1)}\omega^{(i+\lfloor\frac{N}{2}\rfloor)m_j}+\displaystyle\sum\limits_{l=2}^{\lfloor\frac{N}{2}\rfloor-k_{0}+1}\epsilon\widehat{\textbf{z}}_{k_{0}-1+l}\widehat{\textbf{w}}_{i+\lfloor\frac{N}{2}\rfloor-1+l}^{(1)}\omega^{(i+\lfloor\frac{N}{2}\rfloor-1+l)m_j}}{\widehat{\textbf{w}}_{i+\lfloor\frac{N}{2}\rfloor-1}^{(1)}\omega^{(i+\lfloor\frac{N}{2}\rfloor-1)m_j}}.
%v_{2,k}:=\frac{\displaystyle\sum\limits_{l=1}^{\lfloor\frac{N}{2}\rfloor-k}\epsilon\widehat{\textbf{z}}_{k+l}\widehat{\textbf{w}}_{i+\lfloor\frac{N}{2}\rfloor-1+l}^{(1)}\omega^{(i+\lfloor\frac{N}{2}\rfloor-1+l)m_2}}{\widehat{\textbf{w}}_{i+\lfloor\frac{N}{2}\rfloor-1}^{(1)}\omega^{(i+\lfloor\frac{N}{2}\rfloor-1)m_2}},\\
%v_{3,k}:=\frac{\displaystyle\sum\limits_{l=1}^{\lfloor\frac{N}{2}\rfloor-k}\epsilon\widehat{\textbf{z}}_{k+l}\widehat{\textbf{w}}_{i+\lfloor\frac{N}{2}\rfloor-1+l}^{(1)}\omega^{(i+\lfloor\frac{N}{2}\rfloor-1+l)m_3}}{\widehat{\textbf{w}}_{i+\lfloor\frac{N}{2}\rfloor-1}^{(1)}\omega^{(i+\lfloor\frac{N}{2}\rfloor-1)m_3}}.
%\end{aligned}
%\right.
\end{equation}
Define
\begin{align}
\notag f(\widehat{\textbf{z}}_{k_{0}})&:=\frac{v_{1,k_{0}-1}-v_{2,k_{0}-1}}{v_{1,k_{0}-1}-v_{3,k_{0}-1}}\\
&=\dfrac{a\widehat{\textbf{z}}_{k_{0}}+b}{c\widehat{\textbf{z}}_{k_{0}}+d},
\end{align}
%\nonumber\\&=
%%%\dfrac{\frac{\widehat{z}_{k+1}\widehat{w}_{i+N-B+1}^{(l_1)}\omega^{m_1}+\sum\limits_{l=2}^{\lfloor\frac{N}{2}\rfloor-k}\widehat{z}_{k+l}\widehat{w}_{i+N-B+l}^{(l_1)}\omega^{lm_1} }{\widehat{w}_{i+N-B}^{(l_1)}}-\frac{\widehat{z}_{k+1}\widehat{w}_{i+N-B+1}^{(l_2)}\omega^{m_2}+\sum\limits_{l=2}^{\lfloor\frac{N}{2}\rfloor-k}\widehat{z}_{k+l}\widehat{w}_{i+N-B+l}^{(l_2)}\omega^{lm_2} }{\widehat{w}_{i+N-B}^{(l_2)}}}
%%%{\frac{\widehat{z}_{k+1}\widehat{w}_{i+N-B+1}^{(l_1)}\omega^{m_1}+\sum\limits_{l=2}^{\lfloor\frac{N}{2}\rfloor-k}\widehat{z}_{k+l}\widehat{w}_{i+N-B+l}^{(l_1)}\omega^{lm_1} }{\widehat{w}_{i+N-B}^{(l_1)}}-\frac{\widehat{z}_{k+1}\widehat{w}_{i+N-B+1}^{(l_3)}\omega^{m_3}+\sum\limits_{l=2}^{\lfloor\frac{N}{2}\rfloor-k}\widehat{z}_{k+l}\widehat{w}_{i+N-B+l}^{(l_3)}\omega^{lm_3} }{\widehat{w}_{i+N-B}^{(l_3)}}} \\
%\dfrac{\widehat{\textbf{w}}_{i+\lfloor\frac{N}{2}\rfloor}^{(1)}(\omega^{m_1}\!-\!\omega^{m_2})\widehat{\textbf{z}}_{k+1} \!+\!\displaystyle\sum\limits_{l=2}^{\lfloor\frac{N}{2}\rfloor-k}\widehat{\textbf{z}}_{k+l}\widehat{\textbf{w}}_{i+\lfloor\frac{N}{2}\rfloor-1+l}^{(1)}(\omega^{lm_1}\!-\!\omega^{lm_2})} {\widehat{\textbf{w}}_{i+\lfloor\frac{N}{2}\rfloor}^{(1)}(\omega^{m_1}\!-\!\omega^{m_3})\widehat{\textbf{z}}_{k+1} \!+\!\displaystyle\sum\limits_{l=2}^{\lfloor\frac{N}{2}\rfloor-k}\widehat{\textbf{z}}_{k+l}\widehat{\textbf{w}}_{i+\lfloor\frac{N}{2}\rfloor-1+l}^{(1)}(\omega^{lm_1}\!-\!\omega^{lm_3})}.
where
\begin{align}\label{up1133}
\left\{\begin{array}{lll}
a=\epsilon\widehat{\textbf{w}}_{i+\lfloor\frac{N}{2}\rfloor}^{(1)}(\omega^{m_1}-\omega^{m_2}),\\
b=\sum\limits_{l=2}^{\lfloor\frac{N}{2}\rfloor-k_{0}+1}\epsilon\widehat{\textbf{z}}_{k_{0}-1+l}\widehat{\textbf{w}}_{i+\lfloor\frac{N}{2}\rfloor-1+l}^{(1)}(\omega^{lm_1}-\omega^{lm_2}),\\
c=\epsilon\widehat{\textbf{w}}_{i+\lfloor\frac{N}{2}\rfloor}^{(1)}(\omega^{m_1}-\omega^{m_3}),\\
d=\sum\limits_{l=2}^{\lfloor\frac{N}{2}\rfloor-k_{0}+1}\epsilon\widehat{\textbf{z}}_{k_{0}-1+l}\widehat{\textbf{w}}_{i+\lfloor\frac{N}{2}\rfloor-1+l}^{(1)}(\omega^{lm_1}-\omega^{lm_3}).
\end{array}\right.
\end{align}
Since $\widehat{\textbf{w}}_{i+\lfloor\frac{N}{2}\rfloor}^{(1)}\neq 0$,  $ac\neq0$. For the generic analytic signal $\textbf{z}$, we have $ad-bc\neq0.$
That is,    $f(\widehat{\textbf{z}}_{k_{0}})$
meets  the requirements  in  Lemma \ref{lemma3.2}.
Then  $\Im[f(\widehat{\textbf{z}}_{k_{0}})]\neq0$.
By Lemma \ref{L}, $\epsilon\widehat{\textbf{z}}_{k_{0}-1}$ can be determined. This completes the proof.
\end{proof}

\begin{rem}\label{disangeyaoqiu}

 The condition
$\widehat{\textbf{w}}_{i+\lfloor\frac{N}{2}\rfloor}^{(1)}\neq0$ in Theorem \ref{999} is also important since if otherwise, then
the   equation system w.r.t
$\widehat{\mathring{\textbf{z}}}_{\lfloor\frac{N}{2}\rfloor-1}$:
$$
\frac{N|\widehat{y}_{N-i,m_j}^{\textbf{w}(1)}|}{|\widehat{\textbf{w}}_{i+\lfloor\frac{N}{2}\rfloor-1}^{(1)}\omega^{(i+\lfloor\frac{N}{2}\rfloor-1)m_j}|} =\big|\widehat{\mathring{\textbf{z}}}_{\lfloor\frac{N}{2}\rfloor-1}+\frac{\widehat{\textbf{z}}_{\lfloor\frac{N}{2}\rfloor}\widehat{\textbf{w}}_{i+\lfloor\frac{N}{2}\rfloor}^{(1)}\omega^{(i+\lfloor\frac{N}{2}\rfloor)m_j}}{\widehat{\textbf{w}}_{i+\lfloor\frac{N}{2}\rfloor-1}^{(1)}\omega^{(i+\lfloor\frac{N}{2}\rfloor-1)m_j}}\big|,j=1,2,3
$$
%\begin{equation}\label{21229}
%\left\{
%\begin{aligned}
%\frac{N|\widehat{y}_{N-i,m_1}^{\textbf{w}(1)}|}{|\widehat{\textbf{w}}_{i+\lfloor\frac{N}{2}\rfloor-1}^{(1)}\omega^{(i+\lfloor\frac{N}{2}\rfloor-1)m_1}|} &=\big|\widehat{\textbf{z}}_{\lfloor\frac{N}{2}\rfloor-1}+v_{1,\lfloor\frac{N}{2}\rfloor-1}\big|,\\
%\frac{N|\widehat{y}_{N-i,m_2}^{\textbf{w}(1)}|}{|\widehat{\textbf{w}}_{i+\lfloor\frac{N}{2}\rfloor-1}^{(1)}\omega^{(i+\lfloor\frac{N}{2}\rfloor-1)m_2}|} &=\big|\widehat{\textbf{z}}_{\lfloor\frac{N}{2}\rfloor-1}+v_{2,\lfloor\frac{N}{2}\rfloor-1}\big|,\\
%\frac{N|\widehat{y}_{N-i,m_3}^{\textbf{w}(1)}|}{|\widehat{\textbf{w}}_{i+\lfloor\frac{N}{2}\rfloor-1}^{(1)}\omega^{(i+\lfloor\frac{N}{2}\rfloor-1)m_3}|} &=\big|\widehat{\textbf{z}}_{\lfloor\frac{N}{2}\rfloor-1}+v_{3,\lfloor\frac{N}{2}\rfloor-1}\big|,
%\end{aligned}
%\right.
%\end{equation}
degenerates to
$$
%\left\{
%\begin{aligned}
\frac{N|\widehat{y}_{N-i,m_j}^{\textbf{w}(1)}|}{|\widehat{\textbf{w}}_{i+\lfloor\frac{N}{2}\rfloor-1}^{(1)}\omega^{(i+\lfloor\frac{N}{2}\rfloor-1)m_j}|} =\big|\widehat{\mathring{\textbf{z}}}_{\lfloor\frac{N}{2}\rfloor-1}\big|,j=1,2,3.
%\frac{N|\widehat{y}_{N-i,m_2}^{\textbf{w}(1)}|}{|\widehat{\textbf{w}}_{i+\lfloor\frac{N}{2}\rfloor-1}^{(1)}\omega^{(i+\lfloor\frac{N}{2}\rfloor-1)m_2}|} &=\big|\widehat{\textbf{z}}_{\lfloor\frac{N}{2}\rfloor-1}\big|,\\
%\frac{N|\widehat{y}_{N-i,m_3}^{\textbf{w}(1)}|}{|\widehat{\textbf{w}}_{i+\lfloor\frac{N}{2}\rfloor-1}^{(1)}\omega^{(i+\lfloor\frac{N}{2}\rfloor-1)m_3}|} &=\big|\widehat{\textbf{z}}_{\lfloor\frac{N}{2}\rfloor-1}\big|,
%\end{aligned}
%\right.
$$
Clearly, the above  system  is underdetermined  and $\widehat{\textbf{z}}_{\lfloor\frac{N}{2}\rfloor-1}$ can not be recovered exactly.
\end{rem}
The following provides
a design for the windows in Theorem \ref{999}.
\begin{exam}
Choose a $(\lceil\frac{N}{2}\rceil+1)$-bandlimited window
$\widehat{\textbf{w}}^{(1)}$ such that \eqref{tiaojian1} holds with $i\in  \{0, \ldots ,N-1\}$.
Consequently, $\widehat{\textbf{w}}^{(1)}_{i+\lfloor\frac{N}{2}\rfloor}\neq0.$
Now  choose
the other three  $(\lceil\frac{N}{2}\rceil+1)$-bandlimited windows
$\widehat{\textbf{w}}^{(s)}, s=2, 3,4$
such that
$\widehat{\textbf{w}}^{(2)}_{i+\lfloor\frac{N}{2}\rfloor-1}=(\widehat{\textbf{w}}^{(1)}_{i+\lfloor\frac{N}{2}\rfloor-1})^2$,
$\widehat{\textbf{w}}^{(3)}_{i+\lfloor\frac{N}{2}\rfloor-1}=(\widehat{\textbf{w}}^{(1)}_{i+\lfloor\frac{N}{2}\rfloor-1})^3$,
$\widehat{\textbf{w}}^{(4)}_{i+\lfloor\frac{N}{2}\rfloor-1}=(\widehat{\textbf{w}}^{(1)}_{i+\lfloor\frac{N}{2}\rfloor-1})^4$,
$\widehat{\textbf{w}}^{(2)}_{i-1+N}=(\widehat{\textbf{w}}^{(1)}_{i-1+N})^2$,
$\widehat{\textbf{w}}^{(3)}_{i-1+N}=(\widehat{\textbf{w}}^{(1)}_{i-1+N})^3$ and
$\widehat{\textbf{w}}^{(4)}_{i-1+N}=(\widehat{\textbf{w}}^{(1)}_{i-1+N})^4$.
Additionally, it is required that $\widehat{\textbf{w}}^{(1)}_{i+\lfloor\frac{N}{2}\rfloor-1}\neq\widehat{\textbf{w}}^{(1)}_{i-1+N}$, $\widehat{\textbf{w}}^{(1)}_{i+\lfloor\frac{N}{2}\rfloor-1}\overline{\widehat{\textbf{w}}^{(1)}_{i-1+N}}\notin \mathbb{R}$, $|\widehat{\textbf{w}}^{(1)}_{i+\lfloor\frac{N}{2}\rfloor-1}|\neq|\widehat{\textbf{w}}^{(1)}_{i-1+N}|$ and $\widehat{\textbf{w}}^{(1)}_{i+\lfloor\frac{N}{2}\rfloor-1}\widehat{\textbf{w}}^{(1)}_{i-1+N}\neq 0$.
Then $\mathcal{A}_{0}$ in \eqref{JJK} can be expressed as
\begin{align}
\label{JblK}
\begin{array}{lll}
\mathcal{A}_{0}=\left(\begin{array}{llllllll}
|\widehat{\textbf{w}}^{(1)}_{i+\lfloor\frac{N}{2}\rfloor-1}|^2&\widehat{\textbf{w}}^{(1)}_{i+\lfloor\frac{N}{2}\rfloor-1}\overline{\widehat{\textbf{w}}^{(1)}_{i-1+N}}&\overline{\widehat{\textbf{w}}^{(1)}_{i+\lfloor\frac{N}{2}\rfloor-1}}\widehat{\textbf{w}}^{(1)}_{i-1+N}&|\widehat{\textbf{w}}^{(1)}_{i-1+N}|^2\\
|\widehat{\textbf{w}}^{(1)}_{i+\lfloor\frac{N}{2}\rfloor-1}|^4&(\widehat{\textbf{w}}^{(1)}_{i+\lfloor\frac{N}{2}\rfloor-1}\overline{\widehat{\textbf{w}}^{(1)}_{i-1+N}})^2&(\overline{\widehat{\textbf{w}}^{(1)}_{i+\lfloor\frac{N}{2}\rfloor-1}}\widehat{\textbf{w}}^{(1)}_{i-1+N})^2&|\widehat{\textbf{w}}^{(1)}_{i-1+N}|^4\\
|\widehat{\textbf{w}}^{(1)}_{i+\lfloor\frac{N}{2}\rfloor-1}|^6&(\widehat{\textbf{w}}^{(1)}_{i+\lfloor\frac{N}{2}\rfloor-1}\overline{\widehat{\textbf{w}}^{(1)}_{i-1+N}})^3&(\overline{\widehat{\textbf{w}}^{(1)}_{i+\lfloor\frac{N}{2}\rfloor-1}}\widehat{\textbf{w}}^{(1)}_{i-1+N})^3&|\widehat{\textbf{w}}^{(1)}_{i-1+N}|^6\\
|\widehat{\textbf{w}}^{(1)}_{i+\lfloor\frac{N}{2}\rfloor-1}|^8&(\widehat{\textbf{w}}^{(1)}_{i+\lfloor\frac{N}{2}\rfloor-1}\overline{\widehat{\textbf{w}}^{(1)}_{i-1+N}})^4&(\overline{\widehat{\textbf{w}}^{(1)}_{i+\lfloor\frac{N}{2}\rfloor-1}}\widehat{\textbf{w}}^{(1)}_{i-1+N})^4&|\widehat{\textbf{w}}^{(1)}_{i-1+N}|^8\end{array}\right).\end{array}\end{align}
Clearly,
\begin{align}
\begin{array}{lll}
\mathcal{A}_{0}=\left(\begin{array}{llllllll}
1&1&1&1\\
|\widehat{\textbf{w}}^{(1)}_{i+\lfloor\frac{N}{2}\rfloor-1}|^2&\widehat{\textbf{w}}^{(1)}_{i+\lfloor\frac{N}{2}\rfloor-1}\overline{\widehat{\textbf{w}}^{(1)}_{i-1+N}}&\overline{\widehat{\textbf{w}}^{(1)}_{i+\lfloor\frac{N}{2}\rfloor-1}}\widehat{\textbf{w}}^{(1)}_{i-1+N}&|\widehat{\textbf{w}}^{(1)}_{i-1+N}|^2\\
|\widehat{\textbf{w}}^{(1)}_{i+\lfloor\frac{N}{2}\rfloor-1}|^4&(\widehat{\textbf{w}}^{(1)}_{i+\lfloor\frac{N}{2}\rfloor-1}\overline{\widehat{\textbf{w}}^{(1)}_{i-1+N}})^2&(\overline{\widehat{\textbf{w}}^{(1)}_{i+\lfloor\frac{N}{2}\rfloor-1}}\widehat{\textbf{w}}^{(1)}_{i-1+N})^2&|\widehat{\textbf{w}}^{(1)}_{i-1+N}|^4\\
|\widehat{\textbf{w}}^{(1)}_{i+\lfloor\frac{N}{2}\rfloor-1}|^6&(\widehat{\textbf{w}}^{(1)}_{i+\lfloor\frac{N}{2}\rfloor-1}\overline{\widehat{\textbf{w}}^{(1)}_{i-1+N}})^3&(\overline{\widehat{\textbf{w}}^{(1)}_{i+\lfloor\frac{N}{2}\rfloor-1}}\widehat{\textbf{w}}^{(1)}_{i-1+N})^3&|\widehat{\textbf{w}}^{(1)}_{i-1+N}|^6\end{array}\right)\\ \quad\quad\quad\times
\left(\begin{array}{llllllll}
|\widehat{\textbf{w}}^{(1)}_{i+\lfloor\frac{N}{2}\rfloor-1}|^2&0&0&0\\
0&\widehat{\textbf{w}}^{(1)}_{i+\lfloor\frac{N}{2}\rfloor-1}\overline{\widehat{\textbf{w}}^{(1)}_{i-1+N}}&0&0\\
0&0&\overline{\widehat{\textbf{w}}^{(1)}_{i+\lfloor\frac{N}{2}\rfloor-1}}\widehat{\textbf{w}}^{(1)}_{i-1+N}&0\\
0&0&0&|\widehat{\textbf{w}}^{(1)}_{i-1+N}|^2\end{array}\right).
\end{array}
\end{align}
Then   $\mathcal{A}_{0}$   is invertible,
and  the four windows $\textbf{w}^{(s)}, s=1,2,3,4$
meet the requirements
in  Theorem \ref{999}.
As an example for $(N,B,i)=(48,25,25)$,
the graphs of
$\textbf{w}^{(s)}, s=1,2,3,4$
and their DFTs are plotted in Figure \ref{4w}.
\end{exam}

\begin{figure}[htbp] %?2????━o???━??━????????????
\centering %????????????????
\subfigure[$\textbf{w}^{(1)}$]
{\begin{minipage}{7cm} %?2????━o???━??━????????????
\centering
\includegraphics[scale=0.36]{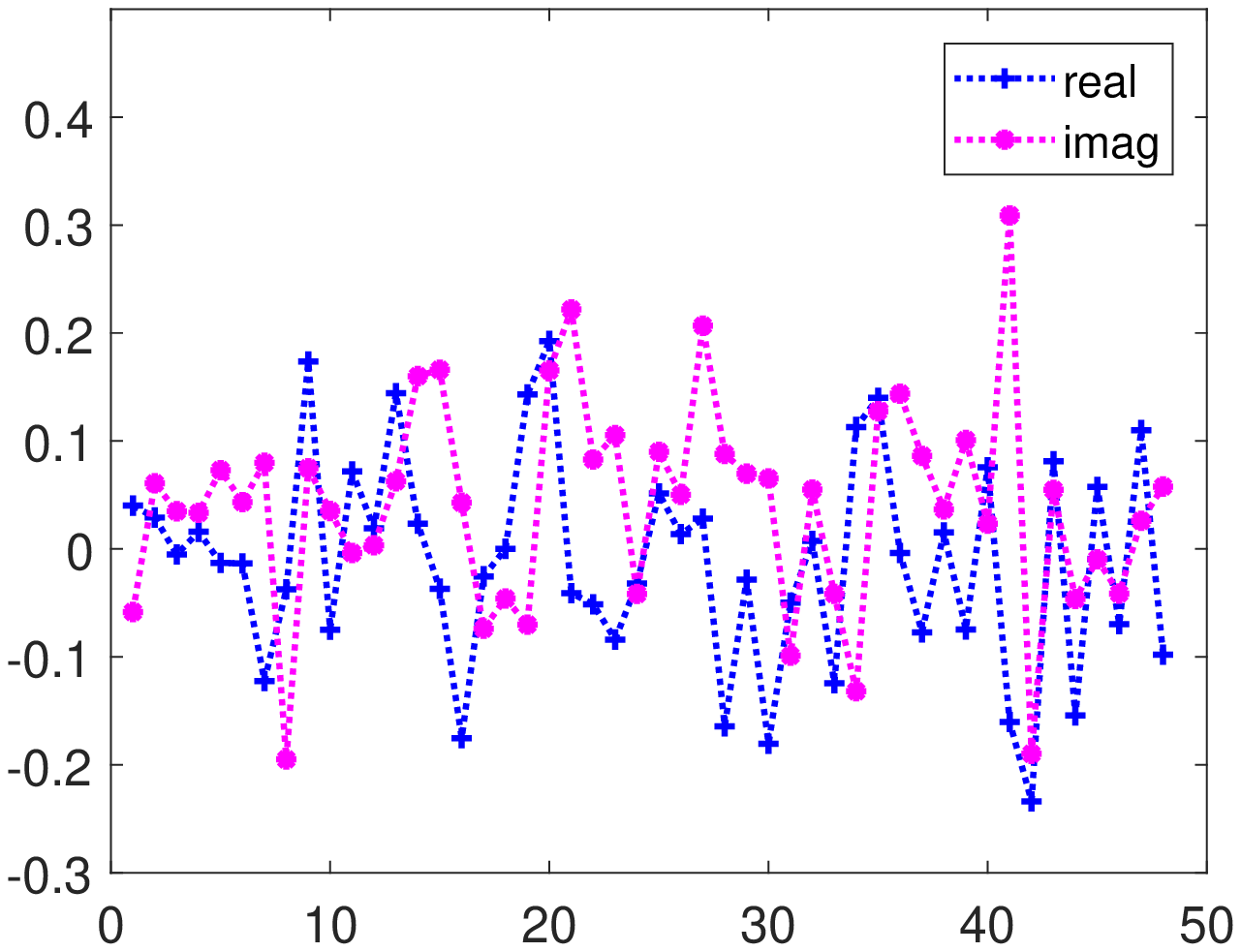}
\end{minipage}}
\subfigure[$\textbf{w}^{(2)}$]
{\begin{minipage}{7cm} %?2????━o???━??━????????????
\centering
\includegraphics[scale=0.36]{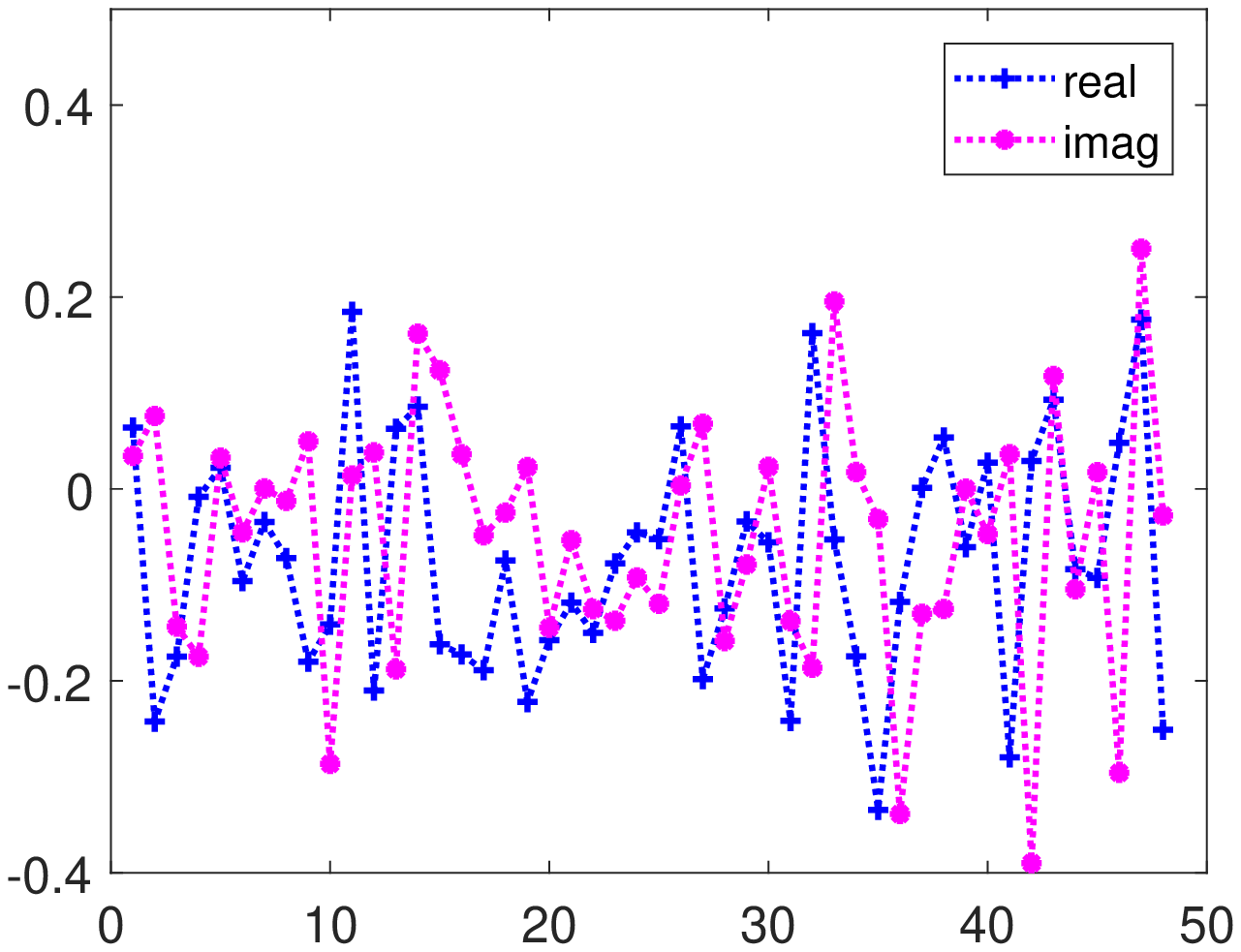}
\end{minipage}}
\subfigure[$\textbf{w}^{(3)}$]
{\begin{minipage}{7cm} %?2????━o???━??━????????????
\centering
\includegraphics[scale=0.36]{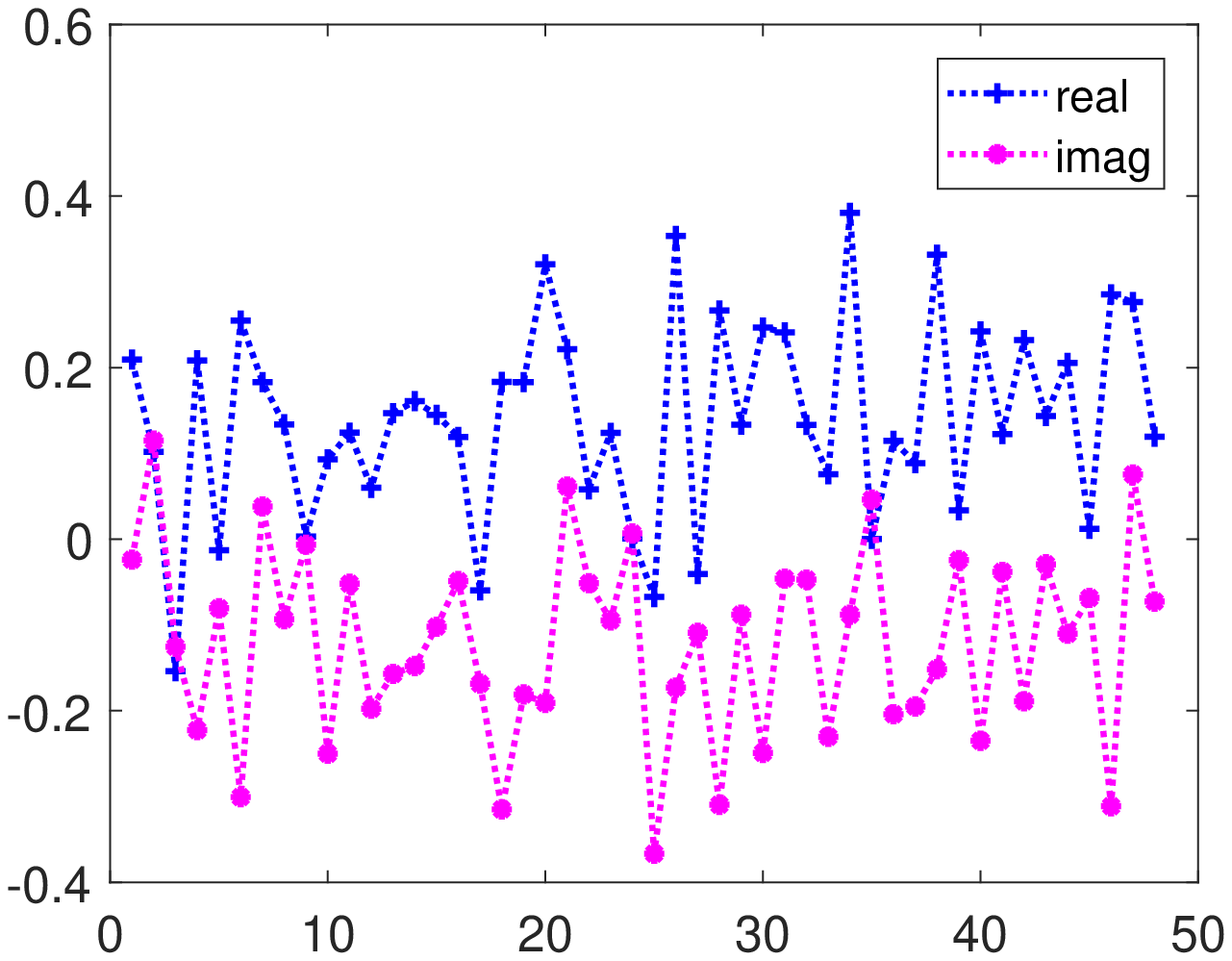}
\end{minipage}}
\subfigure[$\textbf{w}^{(4)}$]
{\begin{minipage}{7cm} %?2????━o???━??━????????????
\centering
\includegraphics[scale=0.36]{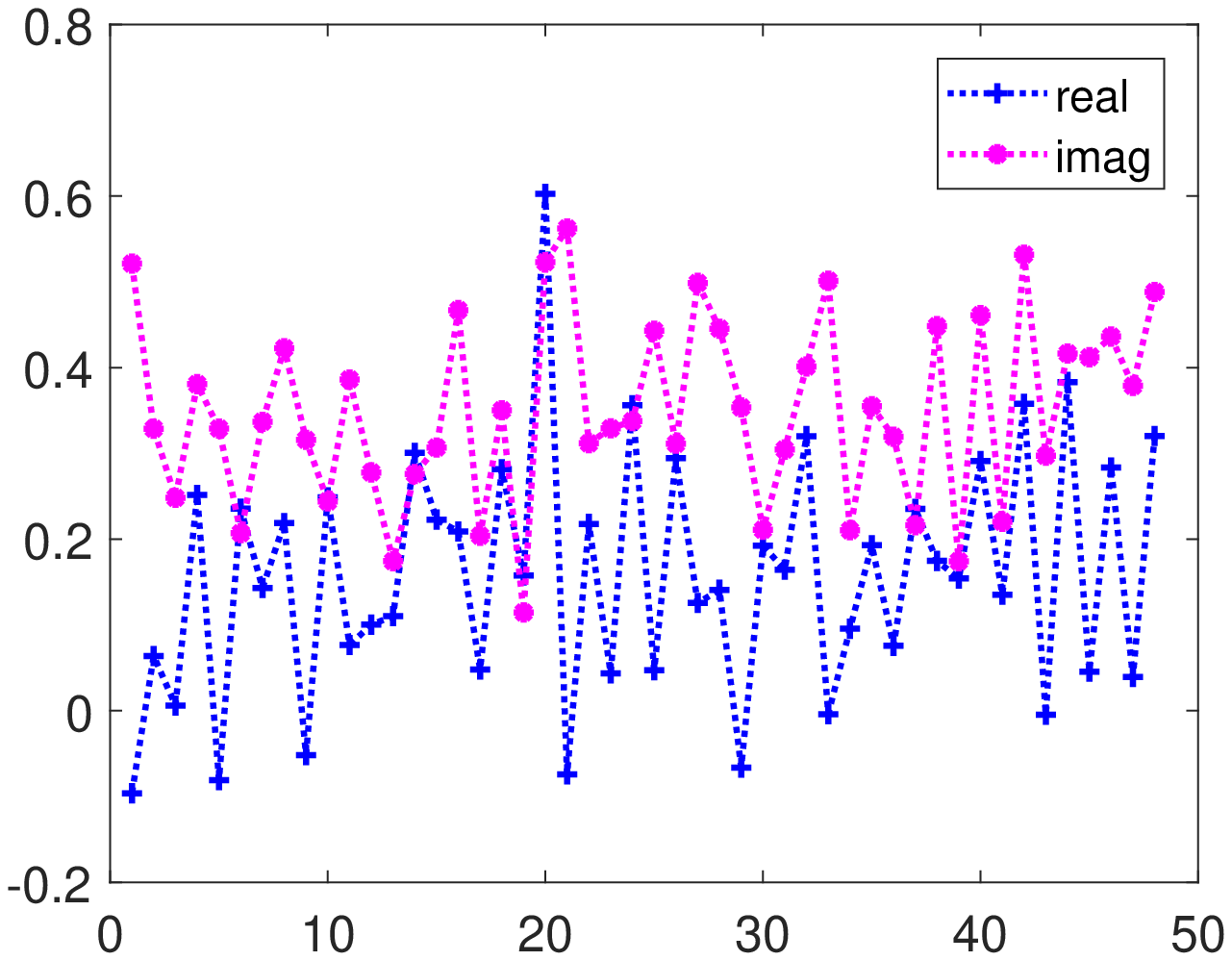}
\end{minipage}}
\subfigure[$\widehat{\textbf{w}}^{(1)}$]
{\begin{minipage}{7cm} %?2????━o???━??━????????????
\centering
\includegraphics[scale=0.36]{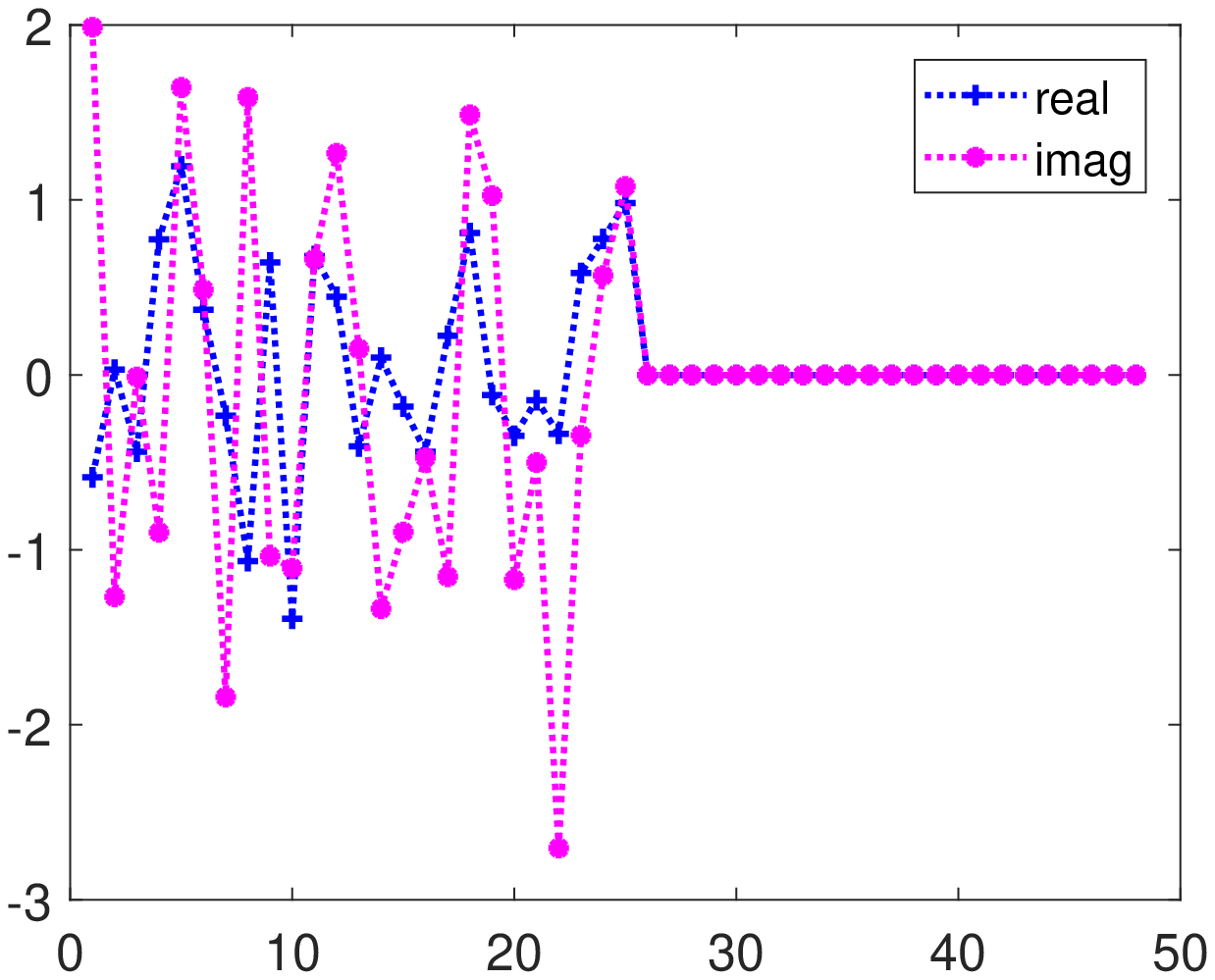}
\end{minipage}}
\subfigure[$\widehat{\textbf{w}}^{(2)}$]
{\begin{minipage}{7cm} %?2????━o???━??━????????????
\centering
\includegraphics[scale=0.36]{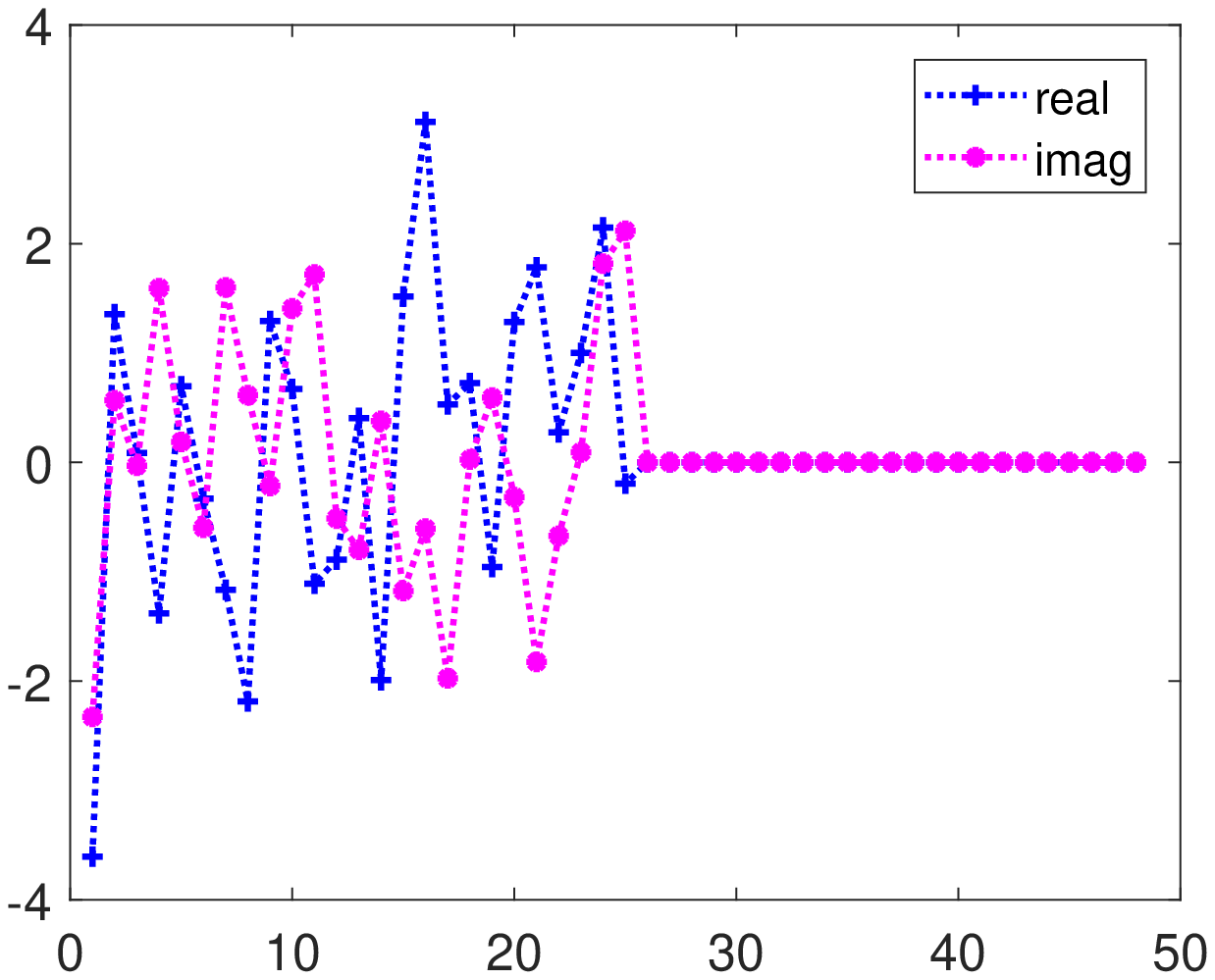}
\end{minipage}}
\subfigure[$\widehat{\textbf{w}}^{(3)}$]
{\begin{minipage}{7cm} %?2????━o???━??━????????????
\centering
\includegraphics[scale=0.36]{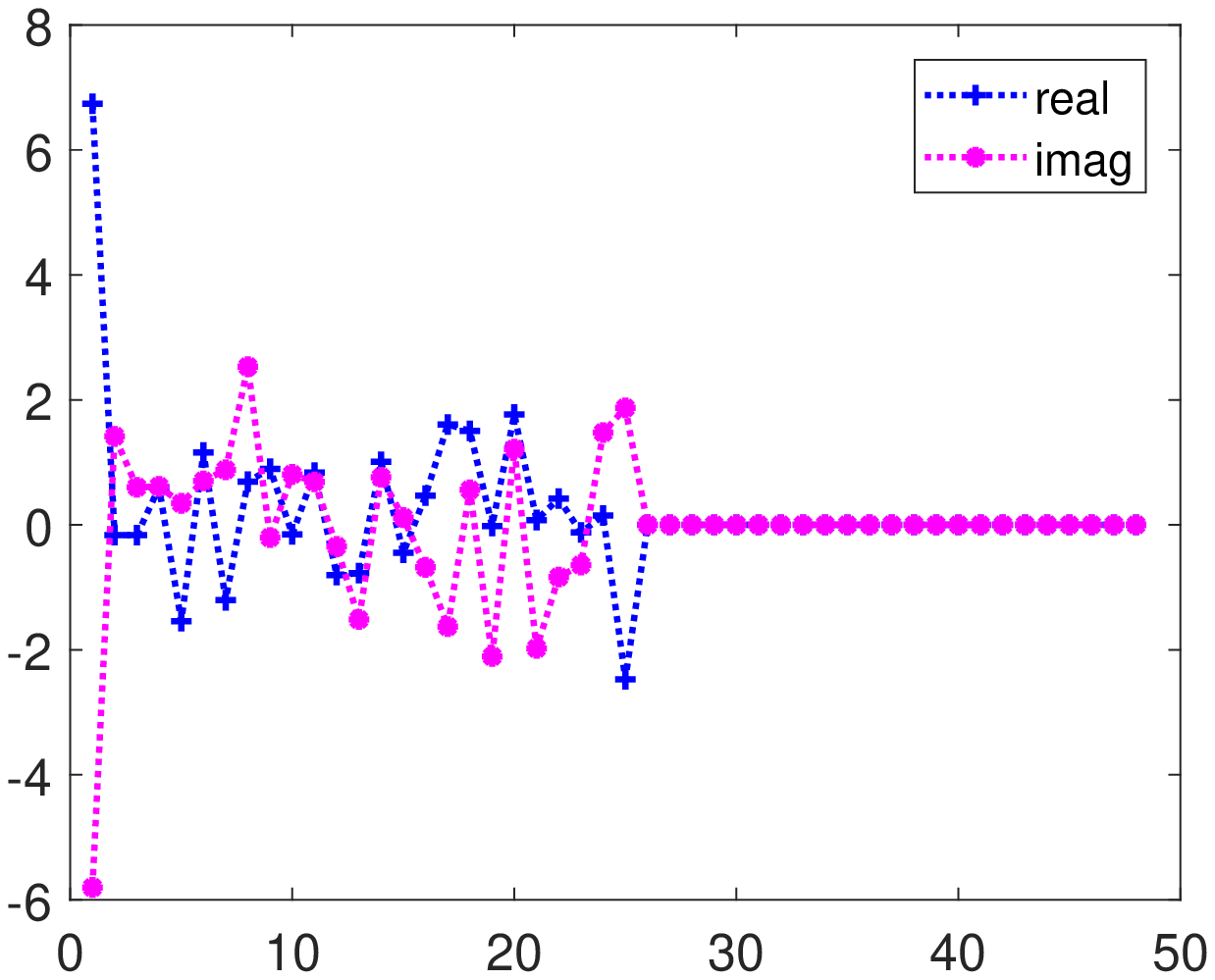}
\end{minipage}}
\subfigure[$\widehat{\textbf{w}}^{(4)}$]
{\begin{minipage}{7cm} %?2????━o???━??━????????????
\centering
\includegraphics[scale=0.36]{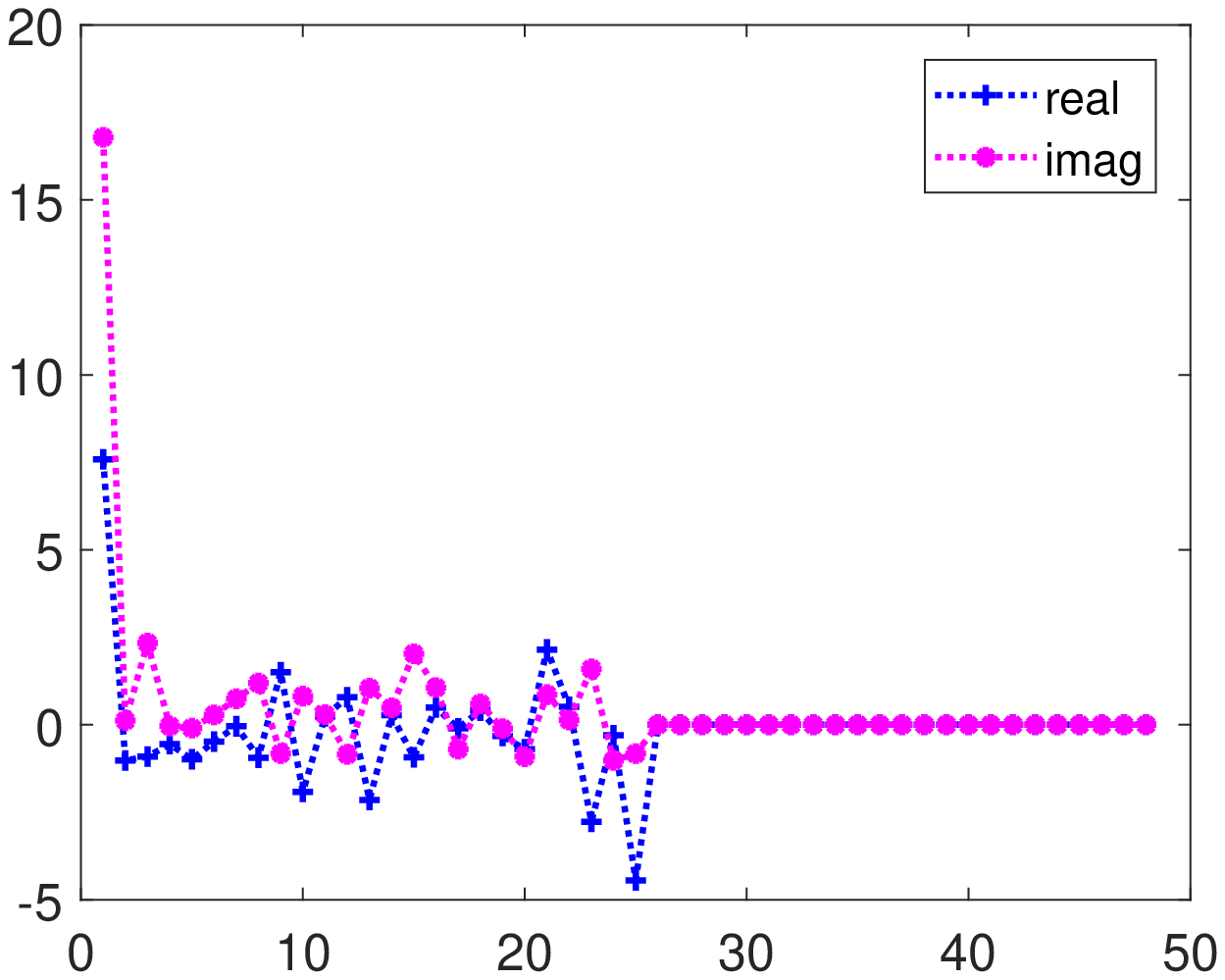}
\end{minipage}}
\caption{(a-d) The real and imaginary parts of the four windows $\textbf{w}^{(s)}$
; (e-h) The real and imaginary parts of  $\widehat{\textbf{w}}^{(s)},s=1,2,3,4$.}
\label{4w}
\end{figure}

%\begin{rem}
%$\big\{|\widehat{y}_{k-(i+\lfloor\frac{N}{2}\rfloor-1),m_j}^{\textbf{w}(s_j)}|:s_j\in\{1,\ldots,4\},0\leqslant m_j\leqslant \lceil N/L\rceil -1 ,j=1,2,3, k=\lfloor\frac{N}{2}\rfloor-1,\ldots,1\big\}$
%
%(i) if $s_1=s_2=s_3$,then $m_1,m_2,m_3$ are pairwise distinct;
%
%(ii) if $s_1\ne s_2=s_3$,then $m_2\ne m_3$;
%
%(iii) if $s_1,s_2,s_3$ are pairwise distinct,then $0\leqslant m_{j}\leqslant \lceil N/L\rceil-1$.
%
%i.e. $|s_1-s_2|+|s_1-s_3|+|s_2-s_3|+|m_1-m_2|+|m_1-m_3|+|m_2-m_3|\ge 3.$
%\end{rem}
\subsection{The third main result:  the analytic window case}\label{thirdmainresult} The main purpose of this subsection is to show that if $N$ is even  and all the windows are analytic,  then fewer measurements than Theorem  \ref{abc} and \ref{999} are  required for the recovery.

\begin{lem}\label{k}
Suppose that $N$ is even, and   $\textbf{z},$ $  \tilde{\textbf{z}}\in \mathbb{C}^N$ are both  generic  analytic signals with  DFTs  $\widehat{\textbf{z}}=(\widehat{\textbf{z}}_0,\widehat{\textbf{z}}_1,\ldots,\widehat{\textbf{z}}_{\frac{N}{2}},0,\ldots,0)$
and $\widehat{\tilde{\textbf{z}}}=(\widehat{\tilde{\textbf{z}}}_0,\widehat{\tilde{\textbf{z}}}_1,\ldots,\widehat{\tilde{\textbf{z}}}_{\frac{N}{2}},0,$
$\ldots,0).$
 Assume that the STFT separation parameter  $0<L<N$ satisfies $\lceil N/L\rceil\geq3$,
and $\textbf{w}^{(1)}$ is an   analytic window   such that $\widehat{\textbf{w}}_1^{(1)}\widehat{\textbf{w}}_0^{(1)}\neq0$.
Let $m_{1}, m_{2}, m_{3}\in \{0, 1, \ldots, \lceil N/L\rceil-1\}$ be  three distinct  parameters.
If $|\widehat{\tilde{\textbf{z}}}_\frac{N}{2}|\ne|\widehat{\textbf{z}}_\frac{N}{2}|$  and $\tilde{\textbf{z}}$
has the same STFT (associted with the window $\textbf{w}^{(1)}$)     magnitudes
as  $\textbf{z}$ at $(\frac{N}{2}-1, m_{j})$, $j=1,2,3$, then
 \begin{align} \label{AAAAAA1} \frac{\widehat{\tilde{\textbf{z}}}^{2}_\frac{N}{2}}{|\widehat{\textbf{z}}_{\frac{N}{2}-1}|^2}
=\frac{(\widehat{\textbf{w}}_0^{(1)})^2}{|\widehat{\textbf{w}}_1^{(1)}|^2}.
\end{align}
\end{lem}
\begin{proof}
By Proposition \ref{jiegou}, both  $\widehat{\textbf{z}}_\frac{N}{2}$ and  $\widehat{\tilde{\textbf{z}}}_\frac{N}{2}$
are real-valued.
Suppose that $\widehat{\tilde{\textbf{z}}}_\frac{N}{2}=\lambda_{\textbf{z},  \tilde{\textbf{z}}}\widehat{\textbf{z}}_\frac{N}{2}$ such that $\pm1\ne \lambda_{\textbf{z},  \tilde{\textbf{z}}} \in\mathbb{R}$.
%We just need to prove that  there exists $m\in \{0, 1, \ldots, \lceil N/L\rceil-1\}$ such that  the STFT magnitudes at $(\frac{N}{2}-1,m)$ of $\textbf{z}$ and $\tilde{\textbf{z}}$ are different  from  each other.
%We arbitrarily choose distinct  $m_{1}, m_{2}, m_{3}\in \big\{0, 1, \ldots, \lceil N/L\rceil-1\big\}$.
 Since the  STFT  magnitudes of $\textbf{z} $ at $(\frac{N}{2}-1, m_{j})$, $j=1,2,3$
are identical to those of  $  \tilde{\textbf{z}}$,
 we have that
\begin{equation}\label{655}\begin{array}{llll}
\frac{1}{N}|\widehat{\textbf{z}}_{\frac{N}{2}-1}\widehat{\textbf{w}}_0^{(1)}+\widehat{\textbf{z}}_\frac{N}{2}\widehat{\textbf{w}}_1^{(1)}\omega^{m_j}|=|\widehat{y}_{\frac{N}{2}-1,m_j}^{\textbf{w}(1)}|=\frac{1}{N}|\widehat{\tilde{\textbf{z}}}_{\frac{N}{2}-1}\widehat{\textbf{w}}_0^{(1)}+\widehat{\tilde{\textbf{z}}}_{\frac{N}{2}}\widehat{\textbf{w}}_1^{(1)}\omega^{m_j}|,
j=1,2, 3.\end{array}\end{equation}
Since $\widehat{\tilde{\textbf{z}}}_\frac{N}{2}=\lambda_{\textbf{z},  \tilde{\textbf{z}}}\widehat{\textbf{z}}_\frac{N}{2}$,
\begin{equation}\label{755}
|\widehat{\textbf{z}}_{\frac{N}{2}-1}\widehat{\textbf{w}}_0^{(1)}+\widehat{\textbf{z}}_{\frac{N}{2}}\widehat{\textbf{w}}_1^{(1)}\omega^{m_j}|^2=|\widehat{\tilde{\textbf{z}}}_{\frac{N}{2}-1}\widehat{\textbf{w}}_0^{(1)}+\lambda_{\textbf{z},  \tilde{\textbf{z}}}\widehat{\textbf{z}}_{\frac{N}{2}}\widehat{\textbf{w}}_1^{(1)}\omega^{m_j}|^2.
\end{equation}
Using Proposition \ref{jiegou}  again, $\widehat{\textbf{w}}_0^{(1)}$ is  real-valued. Then  \eqref{755} is equivalent to
\begin{align}\label{855}\begin{array}{llll}
(\lambda_{\textbf{z},  \tilde{\textbf{z}}}^2-1)\widehat{\textbf{z}}^2_{\frac{N}{2}}|\widehat{\textbf{w}}_1^{(1)}|^2+(|\widehat{\tilde{\textbf{z}}}_{\frac{N}{2}-1}|^2-|\widehat{\textbf{z}}_{\frac{N}{2}-1}|^2)(\widehat{\textbf{w}}_0^{(1)})^2
\\+2\Re\{\overline{\omega^{m_j}\widehat{\textbf{w}}_1^{(1)}}\widehat{\textbf{w}}_0^{(1)}\widehat{\textbf{z}}_{\frac{N}{2}}(\lambda_{\textbf{z},  \tilde{\textbf{z}}}\widehat{\tilde{\textbf{z}}}_{\frac{N}{2}-1}-\widehat{\textbf{z}}_{\frac{N}{2}-1})\}=0.
\end{array}\end{align}
Multiplying by  $\omega^{m_j}$ on both sides of \eqref{855}  leads to
\begin{align}\label{955}\begin{array}{llll}
[(\lambda_{\textbf{z},  \tilde{\textbf{z}}}^2-1)\widehat{\textbf{z}}^2_{\frac{N}{2}}|\widehat{\textbf{w}}_1^{(1)}|^2+(|\widehat{\tilde{\textbf{z}}}_{\frac{N}{2}-1}|^2-|\widehat{\textbf{z}}_{\frac{N}{2}-1}|^2)(\widehat{\textbf{w}}_0^{(1)})^2]\omega^{m_j}\\
+(\lambda_{\textbf{z},  \tilde{\textbf{z}}}\overline{\widehat{\tilde{\textbf{z}}}_{\frac{N}{2}-1}}-\overline{\widehat{\textbf{z}}_{\frac{N}{2}-1})}\widehat{\textbf{w}}_1^{(1)}\widehat{\textbf{w}}_0^{(1)}\widehat{\textbf{z}}_{\frac{N}{2}}\omega^{2m_j}
+(\lambda_{\textbf{z},  \tilde{\textbf{z}}}\widehat{\tilde{\textbf{z}}}_{\frac{N}{2}-1}-\widehat{\textbf{z}}_{\frac{N}{2}-1})\widehat{\textbf{w}}_0^{(1)}\overline{\widehat{\textbf{w}}_1^{(1)}}\widehat{\textbf{z}}_{\frac{N}{2}}=0.
\end{array}\end{align}
Consider the following equation w.r.t $x$:
\begin{align}\label{1055}\begin{array}{llll}
[(\lambda_{\textbf{z},  \tilde{\textbf{z}}}^2-1)\widehat{\textbf{z}}^2_{\frac{N}{2}}|\widehat{\textbf{w}}_1^{(1)}|^2+(|\widehat{\tilde{\textbf{z}}}_{\frac{N}{2}-1}|^2-|\widehat{\textbf{z}}_{\frac{N}{2}-1}|^2)(\widehat{\textbf{w}}_0^{(1)})^2]x\\
+(\lambda_{\textbf{z},  \tilde{\textbf{z}}}\overline{\widehat{\tilde{\textbf{z}}}_{\frac{N}{2}-1}}-\overline{\widehat{\textbf{z}}_{\frac{N}{2}-1})}\widehat{\textbf{w}}_1^{(1)}\widehat{\textbf{w}}_0^{(1)}\widehat{\textbf{z}}_{\frac{N}{2}}x^2
+(\lambda_{\textbf{z},  \tilde{\textbf{z}}}\widehat{\tilde{\textbf{z}}}_{\frac{N}{2}-1}-\widehat{\textbf{z}}_{\frac{N}{2}-1})\widehat{\textbf{w}}_0^{(1)}\overline{\widehat{\textbf{w}}_1^{(1)}}\widehat{\textbf{z}}_{\frac{N}{2}}=0.
\end{array}\end{align}
If the polynomial on the left-hand side of \eqref{1055} is a  non-zero polynomial, then there are at most two solutions to the above equation.
By  \eqref{955}, $\omega^{m_j}, j=1,2,3 $ are the three distinct solutions to \eqref{1055}.
Therefore, all the coefficients in \eqref{1055} are zero. Then
\begin{align}\label{KKKKK}(\lambda_{\textbf{z},  \tilde{\textbf{z}}}\widehat{\tilde{\textbf{z}}}_{\frac{N}{2}-1}-\widehat{\textbf{z}}_{\frac{N}{2}-1})\widehat{\textbf{w}}_0^{(1)}\overline{\widehat{\textbf{w}}_1^{(1)}}\widehat{\textbf{z}}_{\frac{N}{2}}=0
\end{align}
and  \begin{align} \label{KKKKKYU} (\lambda_{\textbf{z},  \tilde{\textbf{z}}}^2-1)\widehat{\textbf{z}}^2_{\frac{N}{2}}|\widehat{\textbf{w}}_1^{(1)}|^2+(|\widehat{\tilde{\textbf{z}}}_{\frac{N}{2}-1}|^2-|\widehat{\textbf{z}}_{\frac{N}{2}-1}|^2)(\widehat{\textbf{w}}_0^{(1)})^2=0.
\end{align}
Since $\textbf{z}$ and $\tilde{\textbf{z}}$ are    generic analytic signals,  we get that  $\widehat{\textbf{z}}_{\frac{N}{2}},
%\widehat{\tilde{\textbf{z}}}_{\frac{N}{2}},
\widehat{\textbf{z}}_{\frac{N}{2}-1},
\widehat{\tilde{\textbf{z}}}_{\frac{N}{2}-1},
\lambda_{\textbf{z},  \tilde{\textbf{z}}} $ are nonzeros.
%Consequently, $\tilde{\textbf{z}}\neq0$.
From \eqref{KKKKK} we have  $\widehat{\tilde{\textbf{z}}}_{\frac{N}{2}-1}=\frac{1}{\lambda_{\textbf{z},  \tilde{\textbf{z}}}}
\widehat{\textbf{z}}_{\frac{N}{2}-1}.$ Combining this   with  \eqref{KKKKKYU} we have that
 \begin{align} \label{hKKKKKYU} (\lambda_{\textbf{z},  \tilde{\textbf{z}}}^2-1)\widehat{\textbf{z}}^2_{\frac{N}{2}}|\widehat{\textbf{w}}_1^{(1)}|^2
 +(\frac{1}{\lambda_{\textbf{z},  \tilde{\textbf{z}}}^{2}}-1)|\widehat{\textbf{z}}_{\frac{N}{2}-1}|^2(\widehat{\textbf{w}}_0^{(1)})^2=0,
\end{align}
 which implies that
 \begin{align} \label{AAAAAA} \widehat{\textbf{z}}^2_{\frac{N}{2}}=
-\frac{(\frac{1}{\lambda_{\textbf{z},  \tilde{\textbf{z}}}^{2}}-1)|\widehat{\textbf{z}}_{\frac{N}{2}-1}|^2(\widehat{\textbf{w}}_0^{(1)})^2}{(\lambda_{\textbf{z},  \tilde{\textbf{z}}}^2-1)|\widehat{\textbf{w}}_1^{(1)}|^2}
=\frac{1}{\lambda_{\textbf{z},  \tilde{\textbf{z}}}^{2}}\frac{|\widehat{\textbf{z}}_{\frac{N}{2}-1}|^2(\widehat{\textbf{w}}_0^{(1)})^2}{|\widehat{\textbf{w}}_1^{(1)}|^2}.
\end{align}
It follows from $\lambda_{\textbf{z},  \tilde{\textbf{z}}}=\frac{\widehat{\tilde{\textbf{z}}}_\frac{N}{2}}{\widehat{\textbf{z}}_\frac{N}{2}}$ and \eqref{AAAAAA} that
 \begin{align} \label{12AAAAAA1} \frac{\widehat{\tilde{\textbf{z}}}^{2}_\frac{N}{2}}{|\widehat{\textbf{z}}_{\frac{N}{2}-1}|^2}
=\frac{(\widehat{\textbf{w}}_0^{(1)})^2}{|\widehat{\textbf{w}}_1^{(1)}|^2},
\end{align}
 which completes the proof.
\end{proof}

 Now we are ready to prove our  third main result.

\begin{theo}\label{3}  Assume that
 $N$ is even and the  STFT separation parameter $L$ satisfies  $\lceil N/L\rceil\geq 3$. Let $m_{1}, m_{2}, m_3\in \{0, 1, \ldots, \lceil N/L\rceil-1\}$ be distinct.
If the  two   windows $\textbf{w}^{(1)}$ and  $\textbf{w}^{(2)}$
 are analytic such that $\widehat{\textbf{w}}_1^{(1)}\widehat{\textbf{w}}_0^{(1)}\neq0$,
$\widehat{\textbf{w}}_{\frac{N}{2}}^{(1)}\widehat{\textbf{w}}_{\frac{N}{2}}^{(2)}\neq0$
 and
$
\widehat{\textbf{w}}_0^{(1)}\widehat{\textbf{w}}_{\frac{N}{2}}^{(2)}-\widehat{\textbf{w}}_0^{(2)}
\widehat{\textbf{w}}_{\frac{N}{2}}^{(1)}\neq0,
$
then any generic analytic   signal  $\textbf{z}\in\mathbb{C}^N$
can be determined (up to a sign) by
its  $(\frac{3N}{2}-1)$ number of STFT magnitudes
\begin{align} \label{2wmeasuements13} \big\{|\widehat{y}_{\frac{N}{2},0}^{\textbf{w}(1)}|,|\widehat{y}_{\frac{N}{2},0}^{\textbf{w}(2)}|,|\widehat{y}_{k,m_j}^{\textbf{w}(1)}|:k=1,\ldots,\frac{N}{2}-1,
 j=1,2,3\big\}.\end{align}
\end{theo}

\begin{proof}
Since   $\textbf{z}$, $\textbf{w}^{(1)}$ and $\textbf{w}^{(2)}$ are all  analytic,  it follows from  Proposition \ref{jiegou} (i)  that the six numbers $\widehat{\textbf{z}}_0,\widehat{\textbf{z}}_{\frac{N}{2}},\widehat{\textbf{w}}_0^{(1)},\widehat{\textbf{w}}_0^{(2)},\widehat{\textbf{w}}_{\frac{N}{2}}^{(1)}$ and $\widehat{\textbf{w}}_{\frac{N}{2}}^{(2)}$ are all  real-valued.

{\bf Step 1: The determination of  $(\widehat{\textbf{z}}_0, \widehat{\textbf{z}}_{\frac{N}{2}}, \widehat{\textbf{z}}_{\frac{N}{2}-1})$.}

In this step, we  prove that $(\widehat{\textbf{z}}_0, \widehat{\textbf{z}}_{\frac{N}{2}}, \widehat{\textbf{z}}_{\frac{N}{2}-1})$
can be determined, up to a sign, by the five measurements $\{|\widehat{y}_{\frac{N}{2},0}^{\textbf{w}(1)}|,|\widehat{y}_{\frac{N}{2},0}^{\textbf{w}(2)}|,|\widehat{y}_{\frac{N}{2}-1,m_j}^{\textbf{w}(1)}|: j=1,2, 3\}$.
Consider the equation system w.r.t the  variable $(\widehat{\mathring{\textbf{z}}}_0,\widehat{\mathring{\textbf{z}}}_{\frac{N}{2}})\in \mathbb{R}^2$:
\begin{equation}\label{1}
\left\{
\begin{aligned}
|\widehat{y}_{\frac{N}{2},0}^{\textbf{w}(1)}| &=\frac{1}{N}|\widehat{\mathring{\textbf{z}}}_\frac{N}{2}\widehat{\textbf{w}}_0^{(1)}+\widehat{\mathring{\textbf{z}}}_0\widehat{\textbf{w}}_{\frac{N}{2}}^{(1)}|,\\
|\widehat{y}_{\frac{N}{2},0}^{\textbf{w}(2)}| &=\frac{1}{N}|\widehat{\mathring{\textbf{z}}}_\frac{N}{2}\widehat{\textbf{w}}_0^{(2)}+\widehat{\mathring{\textbf{z}}}_0\widehat{\textbf{w}}_{\frac{N}{2}}^{(2)}|.
\end{aligned}
\right.
\end{equation}
It follows from $\widehat{\textbf{w}}_0^{(1)}\widehat{\textbf{w}}_{\frac{N}{2}}^{(2)}-
\widehat{\textbf{w}}_0^{(2)}\widehat{\textbf{w}}_{\frac{N}{2}}^{(1)}\neq0$
that
 the solutions (up to a global sign $\epsilon$)  to \eqref{1} are
\begin{align}\label{000}
(\widehat{\mathring{\textbf{z}}}_0,\widehat{\mathring{\textbf{z}}}_{\frac{N}{2}})=\Big(\frac{N(-|\widehat{y}_{\frac{N}{2},0}^{\textbf{w}(1)}|\widehat{\textbf{w}}_{0}^{(2)}+|\widehat{y}_{\frac{N}{2},0}^{\textbf{w}(2)}|\widehat{\textbf{w}}_0^{(1)})}{\widehat{\textbf{w}}_0^{(1)}\widehat{\textbf{w}}_{\frac{N}{2}}^{(2)}-\widehat{\textbf{w}}_0^{(2)}\widehat{\textbf{w}}_{\frac{N}{2}}^{(1)}},
\frac{N(|\widehat{y}_{\frac{N}{2},0}^{\textbf{w}(1)}|\widehat{\textbf{w}}_{\frac{N}{2}}^{(2)}-|\widehat{y}_{\frac{N}{2},0}^{\textbf{w}(2)}|\widehat{\textbf{w}}_\frac{N}{2}^{(1)})}{\widehat{\textbf{w}}_0^{(1)}\widehat{\textbf{w}}_{\frac{N}{2}}^{(2)}-\widehat{\textbf{w}}_0^{(2)}\widehat{\textbf{w}}_{\frac{N}{2}}^{(1)}}
\Big)
\end{align}
and
\begin{align}\label{010}
(\widehat{\mathring{\textbf{z}}}_0,\widehat{\mathring{\textbf{z}}}_{\frac{N}{2}})=\Big(\frac{N(-|\widehat{y}_{\frac{N}{2},0}^{\textbf{w}(1)}|\widehat{\textbf{w}}_{0}^{(2)}-|\widehat{y}_{\frac{N}{2},0}^{\textbf{w}(2)}|\widehat{\textbf{w}}_0^{(1)})}{\widehat{\textbf{w}}_0^{(1)}\widehat{\textbf{w}}_{\frac{N}{2}}^{(2)}-\widehat{\textbf{w}}_0^{(2)}\widehat{\textbf{w}}_{\frac{N}{2}}^{(1)}},
\frac{N(|\widehat{y}_{\frac{N}{2},0}^{\textbf{w}(1)}|\widehat{\textbf{w}}_{\frac{N}{2}}^{(2)}+|\widehat{y}_{\frac{N}{2},0}^{\textbf{w}(2)}|\widehat{\textbf{w}}_\frac{N}{2}^{(1)})}{\widehat{\textbf{w}}_0^{(1)}\widehat{\textbf{w}}_{\frac{N}{2}}^{(2)}-\widehat{\textbf{w}}_0^{(2)}\widehat{\textbf{w}}_{\frac{N}{2}}^{(1)}}
\Big).
\end{align}

 For any $\widehat{\mathring{\textbf{z}}}_\frac{N}{2}$
 given  through \eqref{000} or \eqref{010}, the following equations w.r.t
$\widehat{\mathring{\textbf{z}}}_{\frac{N}{2}-1}$:
\begin{equation}\label{12fg655}\begin{array}{llll}
|\widehat{y}_{\frac{N}{2}-1,m_j}^{\textbf{w}(1)}|=\frac{1}{N}|\widehat{\mathring{\textbf{z}}}_{\frac{N}{2}-1}\widehat{\textbf{w}}_0^{(1)}+\widehat{\mathring{\textbf{z}}}_\frac{N}{2}\widehat{\textbf{w}}_1^{(1)}\omega^{m_j}|,
j=1,2, 3\end{array}\end{equation}
have a unique solution if and only if
the three circles w.r.t the variable $\widehat{\mathring{\textbf{z}}}_{\frac{N}{2}-1}$:
\begin{align}\label{circle}
\frac{N\big|\widehat{y}_{\frac{N}{2}-1,m_j}^{\textbf{w}(1)}\big|}{\big|\widehat{\textbf{w}}^{(1)}_{0}\big|}=\big|\widehat{\mathring{\textbf{z}}}_{\frac{N}{2}-1}+\frac{\widehat{\mathring{\textbf{z}}}_{\frac{N}{2}}\widehat{\textbf{w}}^{(1)}_1\omega^{m_j}}{\widehat{\textbf{w}}^{(1)}_{0}}\big|,
j=1,2,3 \end{align}
 have only one intersection point.
 We next prove that
 for  the  two choices of $\widehat{\mathring{\textbf{z}}}_{\frac{N}{2}}$ given by \eqref{000} and  \eqref{010}:
\begin{align}\label{kkkker}
\widehat{\mathring{\textbf{z}}}_{\frac{N}{2}}=\frac{N(|\widehat{y}_{\frac{N}{2},0}^{\textbf{w}(1)}|\widehat{\textbf{w}}_{\frac{N}{2}}^{(2)}-|\widehat{y}_{\frac{N}{2},0}^{\textbf{w}(2)}|\widehat{\textbf{w}}_\frac{N}{2}^{(1)})}{\widehat{\textbf{w}}_0^{(1)}\widehat{\textbf{w}}_{\frac{N}{2}}^{(2)}-\widehat{\textbf{w}}_0^{(2)}\widehat{\textbf{w}}_{\frac{N}{2}}^{(1)}} \ \hbox{and} \
\widehat{\mathring{\textbf{z}}}_{\frac{N}{2}}=\frac{N(|\widehat{y}_{\frac{N}{2},0}^{\textbf{w}(1)}|\widehat{\textbf{w}}_{\frac{N}{2}}^{(2)}+|\widehat{y}_{\frac{N}{2},0}^{\textbf{w}(2)}|\widehat{\textbf{w}}_\frac{N}{2}^{(1)})}{\widehat{\textbf{w}}_0^{(1)}\widehat{\textbf{w}}_{\frac{N}{2}}^{(2)}-\widehat{\textbf{w}}_0^{(2)}\widehat{\textbf{w}}_{\frac{N}{2}}^{(1)}},\end{align}
there is only one choice  such that the corresponding three circles in \eqref{circle}
have only one intersection point.
By  Lemma \ref{k}, we just need to prove the two aspects:
(1) the two numbers in \eqref{kkkker}  do not have the same absolute values;
(2) Lemma \ref{k} \eqref{AAAAAA1} does not hold.

If (1) does not hold then
 \begin{align} \label{pdtj}  |\widehat{y}_{\frac{N}{2},0}^{\textbf{w}(1)}|\widehat{\textbf{w}}_{\frac{N}{2}}^{(2)}=0 \  \hbox{or} \ |\widehat{y}_{\frac{N}{2},0}^{\textbf{w}(2)}|
\widehat{\textbf{w}}_{\frac{N}{2}}^{(1)}=0.\end{align}
For the generic analytic  signal $\textbf{z}$,
it follows from \eqref{1} that $|\widehat{y}_{\frac{N}{2},0}^{\textbf{w}(1)}|\neq0$
and $|\widehat{y}_{\frac{N}{2},0}^{\textbf{w}(2)}|\neq0.$
This combining with  $\widehat{\textbf{w}}_{\frac{N}{2}}^{(1)}\widehat{\textbf{w}}_{\frac{N}{2}}^{(2)}\neq0$
leads to that  \eqref{pdtj} does not hold. Therefore, (1) hold.

Next we prove
(2).
Without loss of generality, denote
%and \eqref{010}   by
\begin{align}\label{kker}\widehat{\tilde{\textbf{z}}}_\frac{N}{2}=\frac{N(|\widehat{y}_{\frac{N}{2},0}^{\textbf{w}(1)}|\widehat{\textbf{w}}_{\frac{N}{2}}^{(2)}-|\widehat{y}_{\frac{N}{2},0}^{\textbf{w}(2)}|\widehat{\textbf{w}}_\frac{N}{2}^{(1)})}{\widehat{\textbf{w}}_0^{(1)}\widehat{\textbf{w}}_{\frac{N}{2}}^{(2)}-\widehat{\textbf{w}}_0^{(2)}\widehat{\textbf{w}}_{\frac{N}{2}}^{(1)}}, \widehat{\textbf{z}}_\frac{N}{2}=\frac{N(|\widehat{y}_{\frac{N}{2},0}^{\textbf{w}(1)}|\widehat{\textbf{w}}_{\frac{N}{2}}^{(2)}+|\widehat{y}_{\frac{N}{2},0}^{\textbf{w}(2)}|\widehat{\textbf{w}}_\frac{N}{2}^{(1)})}{\widehat{\textbf{w}}_0^{(1)}\widehat{\textbf{w}}_{\frac{N}{2}}^{(2)}-\widehat{\textbf{w}}_0^{(2)}\widehat{\textbf{w}}_{\frac{N}{2}}^{(1)}}.\end{align}
It follows from (1) that   $|\widehat{\textbf{z}}_\frac{N}{2}|\neq|\widehat{\tilde{\textbf{z}}}_\frac{N}{2}|$.
By  \eqref{000}
and \eqref{010},   \eqref{AAAAAA1} is equivalent to
\begin{align}\label{gggggggg}
A_1\widehat{\textbf{z}}_\frac{N}{2}^4+A_2\widehat{\textbf{z}}_\frac{N}{2}^2\widehat{\textbf{z}}_{0}^2+A_3\widehat{\textbf{z}}_\frac{N}{2}^3\widehat{\textbf{z}}_{0}
+A_4\widehat{\textbf{z}}_\frac{N}{2}\widehat{\textbf{z}}_{0}^3
-A_5\widehat{\textbf{z}}_\frac{N}{2}^2|\widehat{\textbf{z}}_{\frac{N}{2}-1}|^2\\\notag
-A_6\widehat{\textbf{z}}_{0}^2|\widehat{\textbf{z}}_{\frac{N}{2}-1}|^2
-A_7\widehat{\textbf{z}}_\frac{N}{2}\widehat{\textbf{z}}_{0}|\widehat{\textbf{z}}_{\frac{N}{2}-1}|^2
+C^2K^4|\widehat{\textbf{z}}_{\frac{N}{2}-1}|^4=0,\end{align}
where all the coefficients $A_{i}$ depend only on $\widehat{\textbf{w}}_0^{(1)},
\widehat{\textbf{w}}_0^{(2)}, \widehat{\textbf{w}}_{\frac{N}{2}}^{(1)}$
and $\widehat{\textbf{w}}_{\frac{N}{2}}^{(2)}$,
and  $
K=\widehat{\textbf{w}}_0^{(1)}\widehat{\textbf{w}}_{\frac{N}{2}}^{(2)}-\widehat{\textbf{w}}_{\frac{N}{2}}^{(1)}\widehat{\textbf{w}}_0^{(2)}$,
$ C=\frac{(\widehat{\textbf{w}}_0^{(1)})^{2}}{|\widehat{\textbf{w}}_1^{(1)}|^2}.$
Clearly, $K\neq0$ and $C\neq0$.
Define a  polynomial
as follows
\begin{align}\label{ggggggg6}
H(x_0,x_1,x_{2},x_{3})=A_1x_{3}^4+A_2x_{3}^2x_{0}^2+A_3x_{3}^3x_{0}
+A_4x_{3}x_{0}^3
-A_5x_{3}^2x_{1}x_{2}\\\notag
-A_6x_{0}^2x_{1}x_{2}
-A_7x_{3}x_{0}x_{1}x_{2}
+C^2K^4x^{2}_{1}x^{2}_{2}.\end{align}
Since $K\neq0$ and $C\neq0$, $H$ is a nonzero polynomial.
%
%It follows from  $\widehat{\textbf{w}}_1^{(1)}\widehat{\textbf{w}}_0^{(1)}\neq0$ and the   matrix in \eqref{xishujuzhen}
%being  invertible that
Replacing $\widehat{\textbf{z}}_k$ by
$\sum_{n=0}^{N-1}\textbf{z}_{n}e^{-2\pi \textbf{i}kn/N},$
then it follows from \eqref{ggggggg6} that  there exists a  polynomial
$\tilde{H}(\textbf{z}_{0}, \ldots, \textbf{z}_{N-1})$
such that $\tilde{H}(\textbf{z}_{0}, \ldots, \textbf{z}_{N-1})=
H(\widehat{\textbf{z}}_0,\widehat{\textbf{z}}_{\frac{N}{2}-1}, \overline{\widehat{\textbf{z}}_{\frac{N}{2}-1}},
\widehat{\textbf{z}}_{\frac{N}{2}}).$ Since $H$ is a nonzero polynomial, $\tilde{H}$
is also a nonzero polynomial.
Moreover, as those of $H$ the coefficients of $\tilde{H}$ depend only on $\widehat{\textbf{w}}_0^{(1)},
\widehat{\textbf{w}}_0^{(2)}, \widehat{\textbf{w}}_{\frac{N}{2}}^{(1)}$
and $\widehat{\textbf{w}}_{\frac{N}{2}}^{(2)}$. Now
it follows from \eqref{gggggggg} that
$\tilde{H}(\textbf{z}_{0}, \ldots, \textbf{z}_{N-1})=
H(\widehat{\textbf{z}}_0,\widehat{\textbf{z}}_{\frac{N}{2}-1}, \overline{\widehat{\textbf{z}}_{\frac{N}{2}-1}},
\widehat{\textbf{z}}_{\frac{N}{2}})=0.$
But for the generic signal
$\textbf{z}$, we have  $\tilde{H}(\textbf{z}_{0}, \ldots, \textbf{z}_{N-1})\neq 0$. This is a contradiction.

Summarizing what addressed above,  (1) and (2) hold. Consequently, only one of two  choices of $\widehat{\textbf{z}}_{\frac{N}{2}}$ in \eqref{kkkker} is feasible. Combining \eqref{000} or \eqref{010},
$(\widehat{\textbf{z}}_0,\widehat{\textbf{z}}_{\frac{N}{2}})$
can be determined up to  a sign.  With $\epsilon(\widehat{\textbf{z}}_0,\widehat{\textbf{z}}_{\frac{N}{2}})$ at hand,
\eqref{12fg655} is equivalent to
\begin{align}\label{circle1}
\frac{N\big|\widehat{y}_{\frac{N}{2}-1,m_j}^{\textbf{w}(1)}\big|}{\big|\widehat{\textbf{w}}^{(1)}_{0}\big|}=\big|\widehat{\mathring{\textbf{z}}}_{\frac{N}{2}-1}+v_{j,\frac{N}{2}-1}\big|,
j=1,2,3 \end{align}
where
$
%\left\{
%\begin{aligned}
v_{j,\frac{N}{2}-1}=\frac{\epsilon\widehat{\textbf{z}}_{\frac{N}{2}}\widehat{\textbf{w}}^{(1)}_1\omega^{m_j}}{\widehat{\textbf{w}}^{(1)}_{0}}.
%v_{2,\lfloor\frac{N}{2}\rfloor-1}:=\frac{\epsilon\widehat{\textbf{z}}_{\lfloor\frac{N}{2}\rfloor}\widehat{\textbf{w}}_{i+\lfloor\frac{N}{2}\rfloor}^{(1)}\omega^{(i+\lfloor\frac{N}{2}\rfloor)m_2}}{\widehat{\textbf{w}}_{i+\lfloor\frac{N}{2}\rfloor-1}^{(1)}\omega^{(i+\lfloor\frac{N}{2}\rfloor-1)m_2}},\\
%v_{3,\lfloor\frac{N}{2}\rfloor-1}:=\frac{\epsilon\widehat{\textbf{z}}_{\lfloor\frac{N}{2}\rfloor}\widehat{\textbf{w}}_{i+\lfloor\frac{N}{2}\rfloor}^{(1)}\omega^{(i+\lfloor\frac{N}{2}\rfloor)m_3}}{\widehat{\textbf{w}}_{i+\lfloor\frac{N}{2}\rfloor-1}^{(1)}\omega^{(i+\lfloor\frac{N}{2}\rfloor-1)m_3}}.
%\end{aligned}
%\right.
$
For the generic analytic signal $\textbf{z}\in \mathbb{C}^N$, we have  $\widehat{\textbf{z}}_{\frac{N}{2}}\neq0$.
Therefore,
$$\dfrac{v_{1,\frac{N}{2}-1}-v_{2,\frac{N}{2}-1}}{v_{1,\frac{N}{2}-1}-v_{3,\frac{N}{2}-1}}\\
=\frac{\omega^{m_1}-\omega^{m_2}}{\omega^{m_1}-\omega^{m_3}}.$$
By  Lemma \ref{m}, we have
$$\Im\Big(\dfrac{v_{1,\frac{N}{2}-1}-v_{2,\frac{N}{2}-1}}{v_{1,\frac{N}{2}-1}-v_{3,\frac{N}{2}-1}}\Big)\ne0.$$
Now it follows from  Lemma \ref{L} that  there exists  a unique solution to the equation system \eqref{circle1} w.r.t $\widehat{\mathring{\textbf{z}}}_{\frac{N}{2}-1}$.
Clearly, $\epsilon\widehat{\textbf{z}}_{\frac{N}{2}-1}$
is a solution to
\eqref{circle1}.
Then $\epsilon\widehat{\textbf{z}}_{\frac{N}{2}-1}$ is the unique solution.
Summarizing what has been  addressed above, from the five  measurements $\{|\widehat{y}_{\frac{N}{2},0}^{\textbf{w}(1)}|,|\widehat{y}_{\frac{N}{2},0}^{\textbf{w}(2)}|,|\widehat{y}_{\frac{N}{2}-1,m_j}^{\textbf{w}(1)}|: j=1,2, 3\}$
the vector $\epsilon(\widehat{\textbf{z}}_0,\widehat{\textbf{z}}_{\frac{N}{2}},\widehat{\textbf{z}}_{\frac{N}{2}-1})$
with $\epsilon\in \{1, -1\}$ can be obtained.

{\bf Step 2: The determination of other  components $
\widehat{\textbf{z}}_1, \ldots, \widehat{\textbf{z}}_{\frac{N}{2}-2}$ }

Having $\epsilon(\widehat{\textbf{z}}_0,\widehat{\textbf{z}}_{\frac{N}{2}},\widehat{\textbf{z}}_{\frac{N}{2}-1})$
at hand, through the similar procedures in the proof of  Theorem \ref{999}, other  components $
\widehat{\textbf{z}}_1, \ldots, \widehat{\textbf{z}}_{\frac{N}{2}-2}$  can be determined
  (up to the sign $\epsilon$) by the $(\frac{3N}{2}-6)$  measurements
\begin{align}\label{212} \Big\{|\widehat{y}_{k,m_j}^{\textbf{w}(1)}| =\frac{1}{N}|\sum_{l=0}^{N-1}\widehat{\textbf{z}}_{k+l}\widehat{\textbf{w}}_l^{(1)}\omega^{lm_j}|: k=1,\ldots,\frac{N}{2}-2, j=1,2,3\Big\}.
\end{align}
This completes the proof.
\end{proof}

\section{Conclusion}
This paper concerns  the phase retrieval of analytic signals in $\mathbb{C}^{N}$ by STFT measurements.
For the window of STFT being bandlimited, we examine the structure of STFT.
In particular, if the windows are $B$-bandlimited our main results state that
a generic analytic signal can be determined up to a sign by $(3\lfloor\frac{N}{2}\rfloor+1)$ measurements.
What is more, if $N$ is even and the windows are also analytic then the above number of measurements can be
reduced to $(3\lfloor\frac{N}{2}\rfloor-1)$.

\bibliographystyle{splncs03}

\bibliography{pls}

\end{document}